\documentclass[letter,scriptaddress,twocolumn, showkeys,aps]{revtex4}
\usepackage{ amssymb }
	\usepackage{amsmath}
	\usepackage{makeidx}
	\usepackage{amsfonts}
        \usepackage{bm}
	\usepackage[ansinew]{inputenc}
	\usepackage[usenames,dvipsnames]{pstricks}
	\usepackage{subfigure}
	\usepackage{epsfig}
	\usepackage{pst-grad} 
	\usepackage{pst-plot} 
		\usepackage{amsthm}
\usepackage{ amssymb }
 	\usepackage{makeidx}
	\usepackage{amsfonts}
 	\usepackage{listings}
 	\usepackage{mathtools}
	\usepackage{float}
    \usepackage{tabularx}
	
	\usepackage[colorlinks,hyperindex]{hyperref}
	\hypersetup
	{
		colorlinks,%
		citecolor=blue,%
		linkcolor=blue,%
		urlcolor=blue,%
	}
	\usepackage{tikz}

\usepackage{amsmath}
\usepackage{amssymb}
\usepackage{graphicx}
\usepackage{xcolor}
\usepackage{float,graphicx}
\usepackage{placeins}
\usepackage[colorinlistoftodos]{todonotes}
\usepackage{cases}

\usepackage{subfigure}
\theoremstyle{definition}
	\newtheorem{definition}{Definition}
    \newtheorem{example}{Example}
	
	\theoremstyle{plain}
	\newtheorem{theorem}{Theorem}
	
	\newtheorem{prop}{Proposition}

 \newcommand{\rowvec}[1]{\ensuremath{\begin{pmatrix}#1\end{pmatrix}}}
 
\makeindex
\begin{document}
\def\definitionautorefname{Definition}
\def\sectionautorefname{Section}
\def\subsectionautorefname{Subsection}
\def\subsubsectionautorefname{Subsubsection}
\def\lemmaautorefname{Lemma}
\def\propautorefname{Proposition}
\def\theoremautorefname{Theorem}
\def\remarkautorefname{Remark}

\title{Spatial Super-Infection and Co-Infection Dynamics in Networks}

\author{Alyssa Yu$^{a}$, Laura P. Schaposnik$^{b}$} 

\begin{abstract}
Understanding interactions between the spread of multiple pathogens during an epidemic is crucial to assessing the severity of infections in human, animal, and plant communities. In this paper, we introduce two new {\it Multiplex Bi-Virus Reaction-Diffusion models (MBRD)} on multiplex metapopulation networks: the super-infection model {\it(MBRD-SI)} and the co-infection model {\it(MBRD-CI)}. These frameworks capture two-pathogen dynamics with spatial diffusion and cross-diffusion, allowing the prediction of infection clustering and large-scale spatial distributions. We establish conditions for Turing and Turing-Hopf instabilities in both models and provide experimental evidence of epidemic pattern formation. Beyond epidemiology, we discuss applications of the {\it MBRD} framework to information propagation, malware diffusion, election forecasting, and urban transportation networks.
\end{abstract}

\maketitle
 
\section{Introduction}

Mathematical models for epidemiology have been crucial to understanding the spread of infections, from Ebola~\cite{kiskowski2016modeling} to malaria~\cite{koella2003epidemiological}. During the COVID-19 pandemic, mathematical models informed policy decisions, including issued public health emergencies, lockdowns, and mask mandates worldwide~\cite{editorial2021epidemiology}. To combat the 2024 measles outbreak in Chicago, Illinois, the Center for Disease Control used a compartmental dynamic model to predict new cases and inform an early response which included mass vaccinations~\cite{masters2024real}.

The field of epidemiology originates from Hippocrates in ancient Greece~\cite{pappas2008insights}, and has evolved significantly since. The first mathematical model for epidemiology was developed by Bernoulli~\cite{dietz2002daniel} to study smallpox spread. Later, in 1927, Kermack and McKendrick introduced the compartmental SIR model~\cite{kermack1927contribution}, in which individuals are separated into the Susceptible, Infected, and Recovery populations. 

Most epidemic models can be categorized as either stochastic or deterministic. There are a number of approaches to stochastic modeling, including Markov chains~\cite{gracy2025modeling}, cellular automata~\cite{sirakoulis2000cellular}, stochastic differential equations~\cite{gray2011stochastic}, branching processes~\cite{mitrofani2021branching}, and percolation~\cite{luo2020minimal}. While most deterministic models are compartmental, modifications can be made to structure them based on factors such as age~\cite{ram2021modified} and risk~\cite{akande2024risk}. Our study is based on the classic SIS model~\cite{hethcote1976qualitative}, in which individuals are compartmentalized into the Susceptible and Infected populations, and individuals become susceptible once again after recovery without lasting immunity.

There are a myriad of studies dedicated to understanding the spread of a single infectious disease. In this paper, we extend the typical SIS framework in the following two ways. 
\begin{itemize}
\item We extend classic SIS models to two-pathogen models, formalized as the {\bf Multiplex Bi-Virus Reaction-Diffusion framework (MBRD)}. Within this, we define the super-infection model  {\bf (MBRD-SI)} and the co-infection model  {\bf (MBRD-CI)}. While we refer to the infecting agents as ``viruses" in this paper, these extensions can also capture the dynamics of interactions between viral strains while they spread across populations.

\item We consider the spatial distribution of infections across a network of populations, which can represent towns, cities, or countries, depending on the spatial scale chosen. To do this, we integrate multiplex networks into our model so varying levels of movement between populations are accounted for.
\end{itemize}

The spread of infectious diseases can be characterized by diffusion processes~\cite{gao2022spatiotemporal, thakar2020unfolding}, and reaction-diffusion equations have been used to model the epidemic spread of a single pathogen~\cite{duan2019turing,zhao2025navigating}. Reaction-diffusion dynamics can often simulate the clustering of infections that occurs between physical communities. For example, we observe that clustering occurs in both Figure~\ref{fig:russia} and Figure~\ref{fig:duanws}, showing that modeling with reaction-diffusion systems may explain some aspects of infection spread in the physical world. In this paper, reaction-diffusion mechanics allow us to describe the spatial distribution of infections in two-pathogen models by treating the susceptible state and each infected state as different morphogens.

\begin{figure}[htbp]
    \centering
    \includegraphics[width=0.38\textwidth]{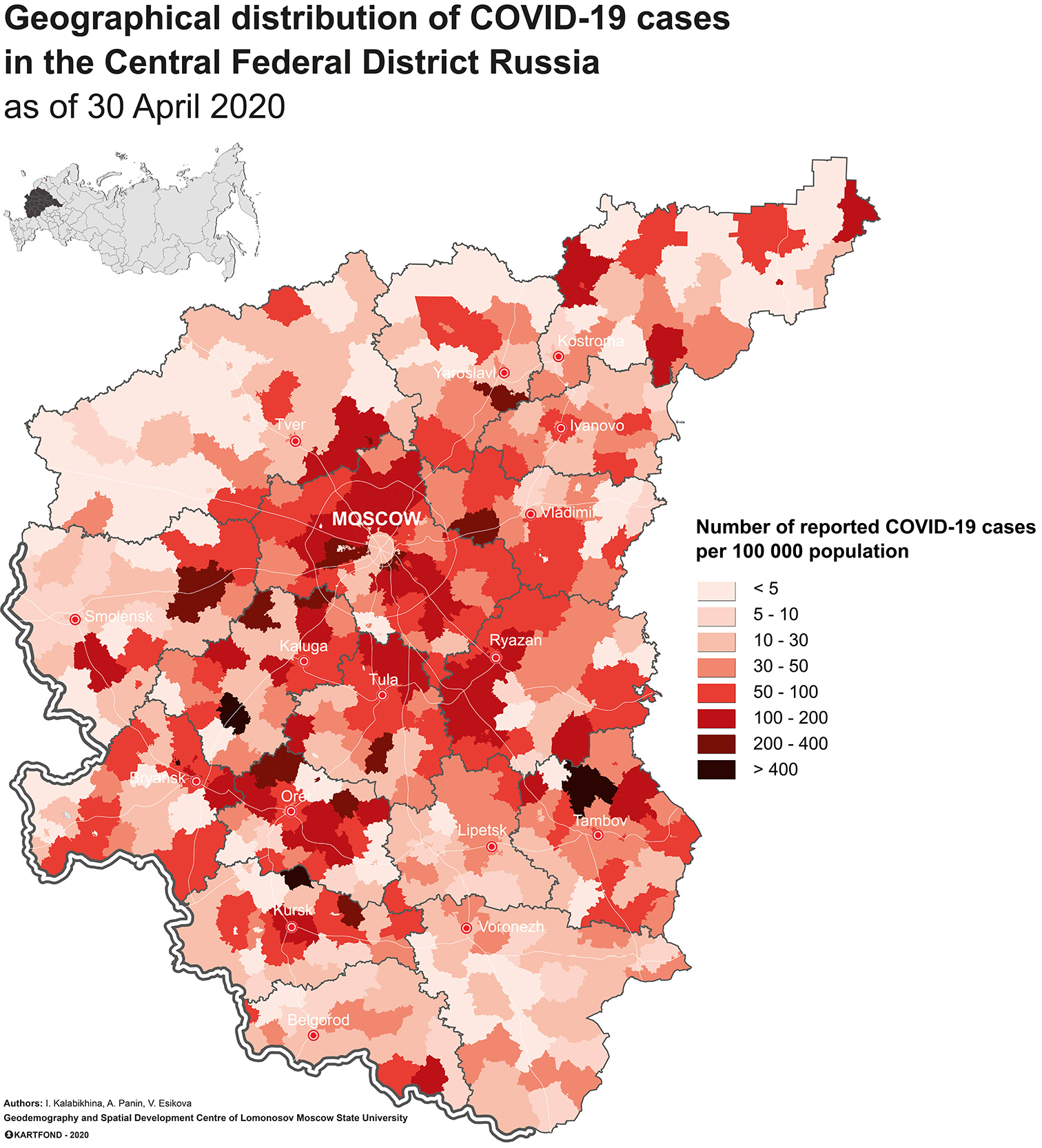}
    \caption{Distribution of COVID-19 infections in the Central Federal District of Russia, from~\cite{kalabikhina2020spatial}.}
    \label{fig:russia}
\end{figure}

\begin{figure}[htbp]
    \centering
    \includegraphics[width=0.45\textwidth]{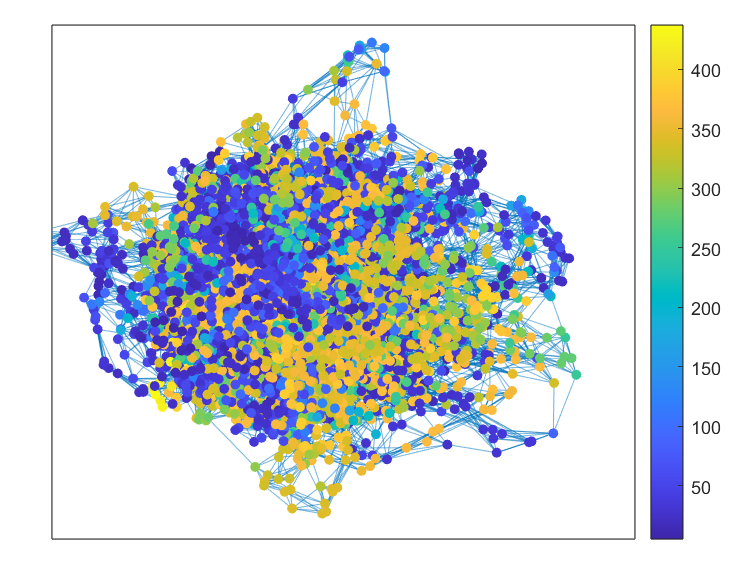}
    \caption{Simulation of infection cases on a Watts-Strogatz network, created using the epidemic model from~\cite{duan2019turing}.}
    \label{fig:duanws}
\end{figure}

To summarize, we make the following three main contributions in this paper:
\begin{itemize}
\item We establish two new \textbf{\textit{ Multiplex Bi-Virus Reaction-Diffusion models} (MBRD)} on multiplex networks: the super-infection model {\bf (MBRD-SI)} and the co-infection model {\bf (MBRD-CI)}. By incorporating reaction-diffusion and cross-diffusion dynamics, these models capture realistic spatial distributions of infections over large geographical ranges and account for complex network structures (see Section~\ref{sec:model}).
\item We perform instability analyses for reaction-diffusion systems with three (resp. four) morphogens on three-layer (resp. four-layer) multiplex networks. This includes explicit conditions for Turing and Turing-Hopf instabilities. To the best of our knowledge, prior work has only addressed reaction-diffusion epidemic systems on two-layer multiplex networks (see Sections~\ref{sec:instability} and~\ref{sec:four-state}).\pagebreak
\item We provide experimental evidence of Turing pattern formation in both MBRD-SI and MBRD-CI, confirming that our theoretical instability conditions give rise to distinct spatial structures (see Section~\ref{sec:examples}).
\end{itemize}

The remainder of this paper is organized as follows. Section~\ref{sec:turing} introduces preliminaries on Turing patterns. Section~\ref{sec:model} introduces novel reaction-diffusion models MBRD-SI and MBRD-CI for super-infection and co-infection.   Section~\ref{sec:instability} establishes instability conditions for three-morphogen systems including the MBRD-SI model, while Section~\ref{sec:four-state} treats the four-morphogen case including the MBRD-CI model. Section~\ref{sec:examples} presents experimental verification of Turing patterns. Section~\ref{sec:applications} discusses potential applications that our framework can be used for. Finally, Section~\ref{sec:conclusion} concludes this paper with a summary of this work and future extensions.



\section{Background}\label{sec:turing}

In 1952, Turing proposed that reaction-diffusion dynamics trigger the formation of many patterns in nature (e.g., the pattern in Figure~\ref{fig:boxfish}). These patterns, known as Turing patterns, are driven by interactions between substances, referred to as morphogens. Subsequently, Gierer and Meinhardt introduced the local autoactivation-lateral inhibition (LALI) framework in 1972, demonstrating that for Turing patterns to form, local self-activation and long-range inhibition must balance each other~\cite{meinhardt2000pattern}. In 1990, Turing patterns were first confirmed experimentally in the chlorite-iodide-malonic acid (CIMA) reaction~\cite{castets1990experimental}. In the following, we first introduce Turing patterns on continuous domains and then Turing patterns on networks. We also discuss previous spatial epidemic models with either one or two pathogens.

\begin{figure}[htbp]
    \centering
    \includegraphics[width=0.40\textwidth]{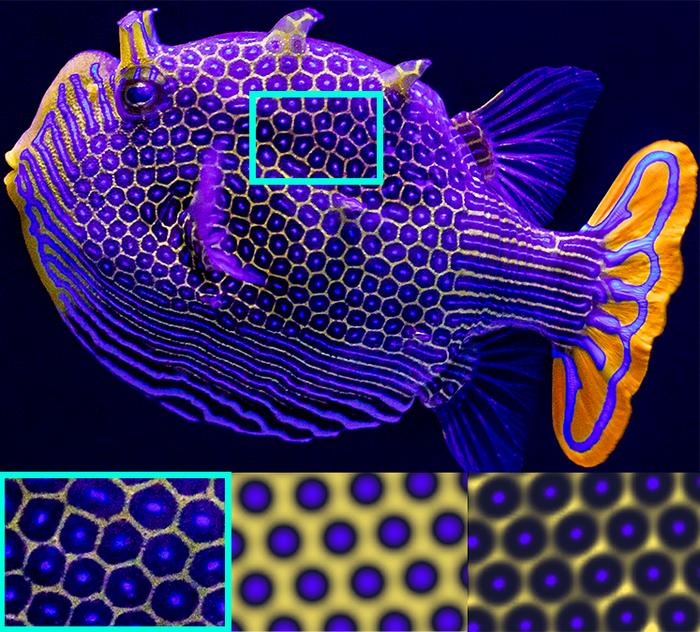}
    \caption{A boxfish (top), a closeup of its pigmentation pattern (bottom-left), and simulations (bottom-center and right), from~\cite{alessio2023diffusiophoresis}.}
    \label{fig:boxfish}
\end{figure}

\subsection{Turing Patterns on Continuous Domains}\label{sec:continuous}

Turing patterns are formed by small fluctuations in the concentration of morphogens, which grow and settle into a spatially organized pattern. The instabilities relevant to this paper are Turing, Hopf, and Turing-Hopf instabilities, which are defined below.
\begin{definition}
    {\em Turing instability} results in a stationary spatial pattern, {\em Hopf instability} results in temporal oscillations only, and {\em Turing-Hopf instability} results in both spatial and temporal oscillations over the same period of time.
\end{definition}

For the most simple example, we consider two reaction-diffusion equations of the form,
\begin{align}\label{eq:reactdiff1}
    \frac{\partial u}{\partial t}&=f(u,v)+D_u\nabla^2 u,\\
    \label{eq:reactdiff2}
    \frac{\partial v}{\partial t}&=g(u,v)+D_v\nabla^2 v,
\end{align}
where $f$ and $g$ describe the reaction kinetics of the morphogens, $D_u$ and $D_v$ are the diffusion coefficients and $\nabla^2 u={\partial^2 u}/{\partial x^2}$ where $x$ is the finite domain $[0,L]$. To obtain a unique pattern, we must also introduce a boundary condition. Most often, the boundary condition imposed is the {\em Neumann condition}, which specifies that there is no flux at the spatial boundary. 

\begin{theorem} [Turing instability conditions for Equations~(\ref{eq:reactdiff1}) and~(\ref{eq:reactdiff2})]~\cite[Equations~(7.13), (7.16), (7.17)]{maini2019turing}

\label{thm:instability}
    Let $f_u\coloneqq{\partial f}/{\partial u}$, $f_v\coloneqq{\partial f}/{\partial v}$, $g_u\coloneqq{\partial g}/{\partial u}$, and $g_v\coloneqq{\partial g}/{\partial v}$. The conditions for Turing instability are
\begin{align*}
    f_u+g_v &< 0, \\
    f_ug_v-f_vg_u &> 0,\\
D_vf_u+D_ug_v&>2\sqrt{D_uD_v(f_ug_v-f_vg_u)},\\
    k_{-}^2&<\left(\frac{n\pi}{L}\right)^2<k_{+}^2,
\end{align*}
where \[k_{\pm}^2=\frac{f_u+g_v\pm\sqrt{(f_u+g_v)^2-4D_uD_v(f_ug_v-f_vg_u)}}{2D_uD_v}.\]
\end{theorem}

Theorem~\ref{thm:instability} is proven by considering a small perturbation $(\hat{u}(x,t),\hat{v}(x,t))$ to the equilibrium state and linearizing the system with multivariable Taylor polynomial expansions.
We then substitute an ansatz solution of the form $\bm{\hat{u}}=\bm{a}\exp(ikx+\lambda(k^2)t)$, where $\bm{\hat{u}}\coloneqq\rowvec{\hat{u}&\hat{v}}^T$, $\bm{a}$ is a constant vector, $k$ is the wave number or the number of spatial oscillations within a certain length, and $\lambda$ is the temporal growth rate. From this, we obtain the characteristic equation $\lambda^2-(f_u+g_v)\lambda+(f_ug_v-f_vg_u)=0$, and with further analysis, we obtain the four conditions above. We will use a similar idea for the instability analysis in this paper.


\subsection{Turing Patterns on Networks}
\label{sec:network}

Most Turing models proposed have been on continuous domains; however, the branch of Turing patterns on complex networks, first introduced in~\cite{nakao2010turing}, has recently become prevalent. 

In the following, we introduce Turing patterns on complex networks. Consider a unweighted network $G\coloneqq(V,E)$ with $|V|=N$, where an edge from node $i$ to node $j$ is denoted by $(i,j)$. We assume here that $G$ is undirected. The entries of the adjacency matrix $\bm{A}(G)$ are defined as 
\[
A^G_{ij} \coloneqq
\left\{
  \begin{array}{ll}
    1, & \text{if } (i,j) \in E, \\
    0, & \text{otherwise}
  \end{array}
\right.
\]

We define a matrix $\bm{L}(G)$ as a function of graph $G$ as follows:
\begin{definition}\label{def:laplacian}
    If $G$ is a graph with $N$ nodes, then we define $\bm{L}(G)$ to be an $N\times N$ matrix with entries
    \[L_{ij}\coloneqq A^G_{ij}-\delta_{ij}k^G_i,\]
    where \[\delta_{ij}\coloneqq 
\begin{cases}
    1, & \text{if $i=j$},\\
    0, & \text{otherwise.}
\end{cases}\] and $A^G_{ij}$ are the entries of the adjacency matrix $\bm{A}(G)$. Moreover
$k^G_i$ is the degree of node $i$ and satisfies $k^G_i\coloneqq\sum_{j=1}^N A^G_{ij}$.
\end{definition}
         
Note that our definition of $\bm{L}$ is the negative of the combinatorial Laplacian. The diffusion of a morphogen from node $j$ to node $i$ is of rate $D_u(u_j-u_i)$. We can add these rates to get the total amount of the morphogen that enters a node. Thus, the amount of substance entering node $i$ is 
\begin{align*}
\dot{u_i}=D_u\sum_{j=1}^n A_{ij}(u_j-u_i)
&=D_u\left(\sum_{j=1}^n A_{ij}u_j\right)-D_uk_iu_i\\
&=D_u\sum_{j=1}^n L_{ij}u_j.
\end{align*}

In a network, a two-morphogen reaction-diffusion system is formulated as
\begin{align*}
    \frac{d u_i}{dt}&=f(u_i,v_i)+D_u\sum_{j=1}^n L_{ij}u_j,\\
    \frac{d v_i}{dt}&=g(u_i,v_i)+D_v\sum_{j=1}^n L_{ij}v_j,
\end{align*}
for all $i=1,2,\dots,n$, where $L_{ij}$ are the entries of the combinatorial Laplacian.

\subsection{Previous Epidemic Models}

In reaction-diffusion epidemic models, the states, such as Susceptible and Infected, are treated as morphogens and the model is a function of the relative densities of those separate populations at every node. To the best of our knowledge, the first SI reaction-diffusion model was introduced by Webb~\cite{webb1981reaction} in 1981. This framework has been extended to analyze the spread of specific pathogens. For example, Bai {\em et al.}~\cite{bai2018reaction} proposed a malaria reaction-diffusion model, accounting for seasonality and incubation. Likewise, Wang {\em et al.}~\cite{wang2024rigorous} studied a similar COVID model accounting for superspreaders and asymptomatic cases on a continuous domain. Recent research has also proposed epidemic models on complex networks. For instance, Duan {\em et al.}~\cite{duan2019turing} investigated a SIS reaction-diffusion model on a single-layer complex network. 

Epidemic models have been studied on two-layer multiplex networks, where layers have the same sets of nodes but can have different connectivity and house different diffusing morphogens representing different aspects of the system. Zhao and Shen introduced a reaction-diffusion epidemic model with $S$ and $I$ states on a two-layer network~\cite{zhao2025navigating} with cross-diffusion, meaning that, when the movement of individuals on the $S$ and $I$ layers induce diffusion on the other layer. Reaction-diffusion two-strain models have also been proposed on continuous domains. Shi and Zhao~\cite{shi2021analysis} analyzed a reaction-diffusion two-strain malaria model and Lu {\em et al.}~\cite{lu2024application} introduced a two-strain COVID model.  

\begin{table}[ht]
\noindent
\begin{tabularx}{0.45\textwidth}{|l|X|}
\hline
\multicolumn{2}{|c|}{\textbf{Parameters for Epidemic Models}} \\
\hline
\textbf{Symbol} & \textbf{Description} \\
\hline
$\beta_1$ & Transmission rate of pathogen 1 from population infected only with pathogen 1\\
$\beta_2$ & Transmission rate of pathogen 2 from population infected only with pathogen 1\\
$\beta_{10}$ & Transmission rate of pathogen 1 only from co-infected population\\
$\beta_{02}$ & Transmission rate of pathogen 2 only from co-infected population\\
$\beta_{12}$ & Transmission rate of co-infection from co-infected population\\
$\gamma_1$ & Recovery rate for pathogen 1 \\
$\gamma_2$ & Recovery rate for pathogen 2 \\
$\alpha_1$ & Virulence of pathogen 1\\
$\alpha_2$ & Virulence of pathogen 2 \\
$r$ & Natural growth of susceptible population\\
$K$ & Maximum environmental capacity density\\
$A$ & Critical spatial carrying capacity density\\
$\mu$ & Natural mortality rate \\
$\sigma$ & Rate of host takeover by the more virulent strain \\
$d_{11}$ & Diffusion rate of susceptible population\\
$d_{12}$ & Cross-diffusion rate induced by movement of population infected with pathogen 1\\
$d_{13}$ & Cross-diffusion rate induced by movement of population infected with pathogen 2\\
$d_{22}$ & Diffusion rate of population infected with pathogen 1\\
$d_{33}$ & Diffusion rate of population infected with pathogen 2\\
\hline
\end{tabularx}
\caption{Parameters used in the epidemic models.}
\label{table:symbol_table}
\end{table}

To the extent of our knowledge, no previous reaction-diffusion models have considered {\bf Multiplex Bi-Virus Reaction-Diffusion frameworks}, including the super-infection case (MBRD-SI) or the co-infection case (MBRD-CI), on discrete domains.

\section{Epidemic Models}\label{sec:model}

In this section, we develop the two {\bf Multiplex Bi-Virus Reaction-Diffusion models}: the super-infection model {\bf (MBRD-SI)} and the co-infection model {\bf (MBRD-CI)}. In this context, super-infection refers to when pathogens cannot coexist in the same host and a more virulent pathogen can ``steal'' the host from a less virulent pathogen. Co-infection describes scenarios when a host can be infected with both viruses at once. In our co-infection model, no pathogen can ``steal'' hosts from the other.

In addition to interactions between two different viruses, our superinfection and co-infection models apply to interactions between different strains of the same virus~\cite{nowak1994superinfection,choisy2010mixed,alizon2013dynamics}. Following conventional notation, we use $S$, $I_1$, $I_2$, and $I_{12}$ in the classic SIS model to represent the densities of the susceptible population, population mono-infected by pathogen $1$, population mono-infected by pathogen $2$, and the co-infected population. We will sometimes refer to a host infected by pathogen $1$ (resp. pathogen $2$) as $I$-infected ($J$-infected).

When discussing the reaction-diffusion equations, we use the symbols $S$, $I$, $J$, and $C$ instead to represent the susceptible, $I$-infected (both mono and co-infections), $J$-infected, and co-infected population densities. Note that Table~\ref{table:symbol_table} introduces the key notation used in Sections~\ref{sec:superinfect} and~\ref{sec:coinfect-model}.

\subsection{The Super-Infection Model (MBRD-SI)}
\label{sec:superinfect}

The following superinfection SIS model was developed by Nowak and May in 1994~\cite{nowak1994superinfection}:
\begin{align*}
    \frac{dS}{dt}&=B-(\mu+\beta_1I_1+\beta_2 I_2)S,\\
    \frac{dI_1}{dt}&=I_1(\beta_1 S-\mu-\alpha_1-\sigma\beta_2 I_2),\\
    \frac{dI_2}{dt} &=I_2(\beta_2S-\mu-\alpha_2+\sigma\beta_2I_1),\\
    1&=S+I_1+I_2,
\end{align*}
where $S$, $I_1$, and $I_2$ represent the proportion of the population that is susceptible, infected by the first strain, and infected by the second strain, respectively. It is assumed that both strains cannot coinfect a single host. Moreover, this model assumes that pathogen $2$ is more virulent than pathogen $1$. Note that $\sigma$ represents the relative rate of superinfection of hosts already infected with pathogen $1$ relative to the transmission of pathogen $2$ to uninfected hosts. When $\sigma>1$, hosts already infected with pathogen $1$ are more likely than uninfected hosts to become infected with pathogen $2$.

Inspired by~\cite{zhao2025navigating}, we incorporate a logistic growth framework to model the growth of the susceptible population, which is suited for reaction dynamics in small communities or cities. Due to factors such as low social capital, populations with low densities will grow relatively slowly. Moreover, due to resource shortages and lower quality of life, populations with high densities will often converge to a carrying capacity, exhibiting the Allee effect.

We adjust this model by adding in recovery rates from both pathogens and modify it so that $S$, $I_1$, and $I_2$ represent the number of individuals or the population densities, as shown in the following system:
\begin{align*}
    \frac{dS}{dt} &= rS\left(1-\frac{S}{K}\right)\left(\frac{S}{A}-1\right)-\frac{(\beta_1I_1+\beta_2 I_2)S}{S+I_1+I_2}\\
    &\quad +\gamma_1 I_1+\gamma_2 I_2-\mu S, \\
\frac{dI_1}{dt} &= I_1\left(\frac{\beta_1 S}{S+I_1+I_2}-\mu-\alpha_1-\gamma_1-\frac{\sigma\beta_2 I_2}{S+I_1+I_2}\right), \\
\frac{dI_2}{dt} &= I_2\left(\frac{\beta_2S}{S+I_1+I_2}-\mu-\alpha_2-\gamma_2+\frac{\sigma\beta_2I_1}{S+I_1+I_2}\right).
\end{align*}

We represent this system with Figure~\ref{fig:superinfect_flowchart}. The susceptible population growth is described by 
$ rS\left(1-\frac{S}{K}\right)\left(\frac{S}{A}-1\right)$. The movement from the susceptible population to the two infected populations are represented by
$\frac{\beta_1I_1S}{S+I_1+I_2}\quad\text{ and }\quad\frac{\beta_2 I_2S}{S+I_1+I_2}$. The movement from the $I$-infected to $J$-infected population due to superinfection is
$\frac{\sigma\beta_2 I_1I_2}{S+I_1+I_2}$. The movement from either infected population to the susceptible population due to recovery is represented by $\gamma_1 I$ and $\gamma_2 J$. Finally, the total population deaths are represented by $\mu S$, $(\mu +\alpha_1)I$, and $(\mu+\alpha_2)J$ for the three populations, respectively.

\begin{figure}[htbp]
    \centering
    \includegraphics[width=0.4\textwidth]{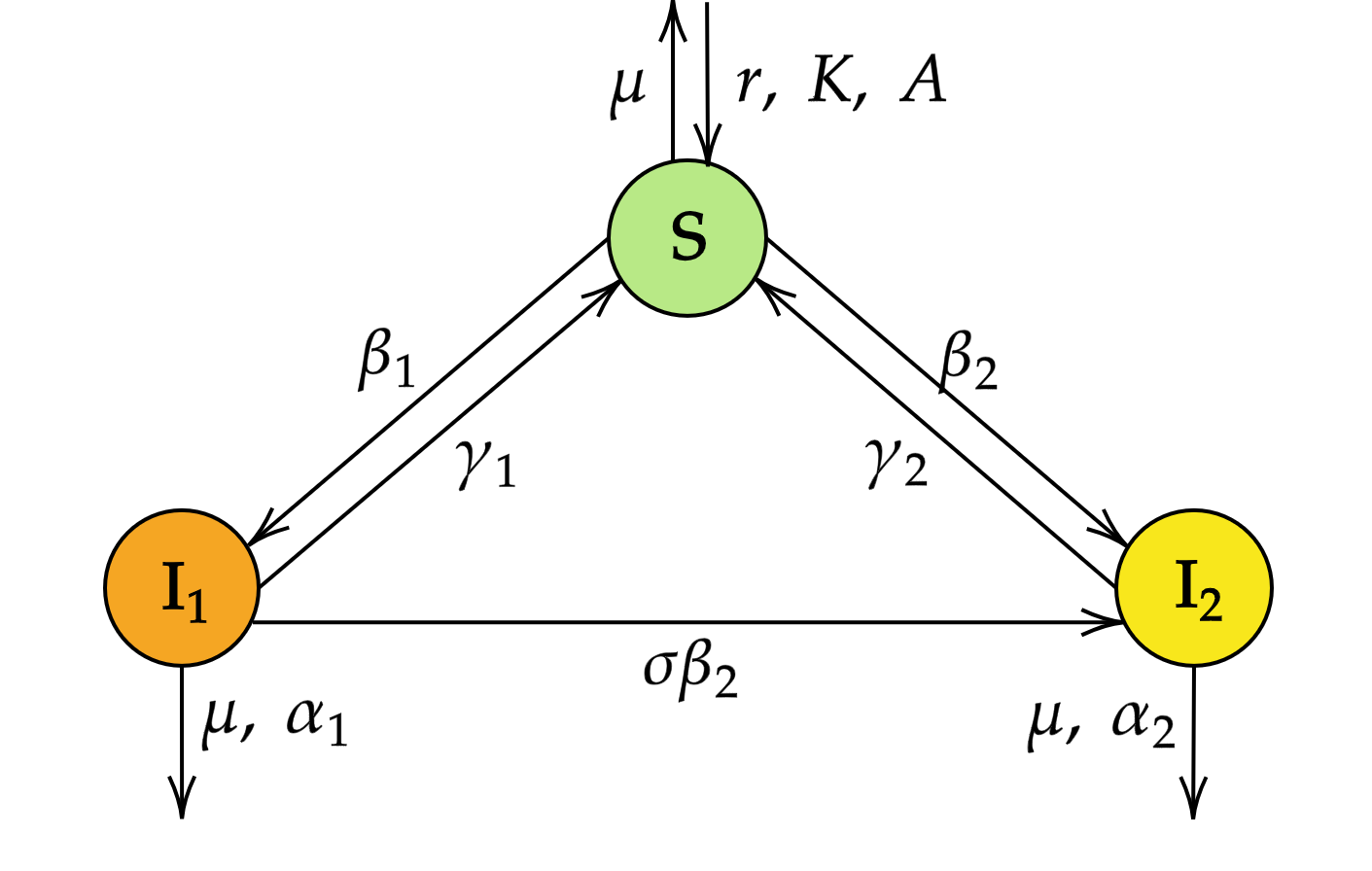}
    \caption{Flowchart for the super-infection model (MBRD-SI).}
    \label{fig:superinfect_flowchart}
\end{figure}

We consider the three-layer multiplex network pictured in Figure~\ref{fig:superinfect_multiplex}, where each layer has the same set of nodes. In an epidemiological context, it is most reasonable for the average degree of the $I$ and $J$ layers to be the same or lower than the average degree of the $S$ layer, as infected people tend to migrate less.
\begin{figure}[htbp]
    \centering
    \includegraphics[width=0.40\textwidth]{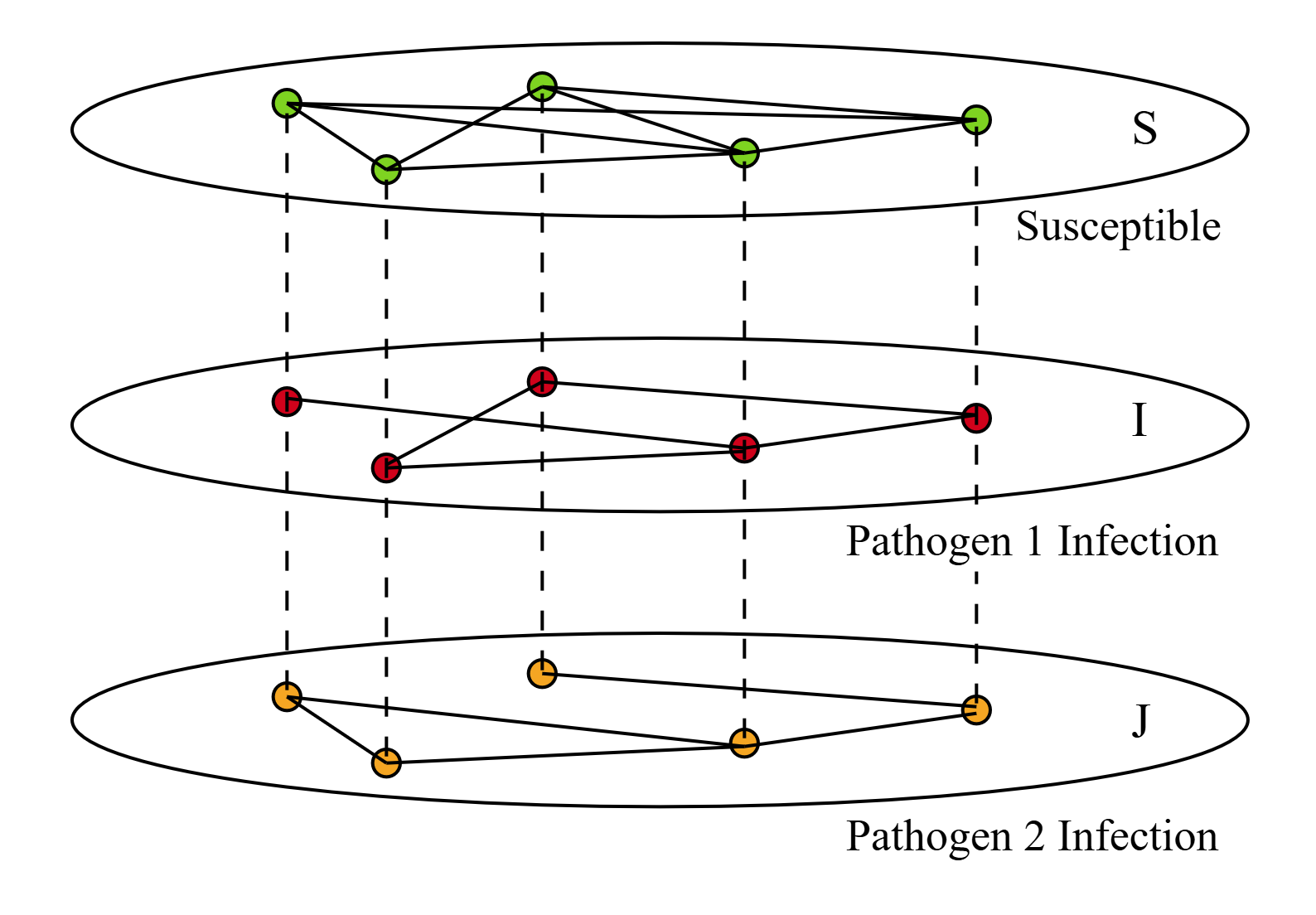}
    \caption{Three-layer multiplex network for MBRD-SI.}
    \label{fig:superinfect_multiplex}
\end{figure}

In order to understand the spatial distribution of infected populations, we treat the $S$, $I$, and $J$ population densities as morphogens which diffuse on the corresponding separate layers of Figure~\ref{fig:superinfect_multiplex}. We let $S_i$, $I_i$, and $J_i$ be the densities of the susceptible population, $I$-infected population, and $J$-infected population, respectively. We denote the $S$, $I$, and $J$ layers to be graphs $G_S$, $G_I$, and $G_J$, respectively. Then, we let $\bm{L}^{(S)}=\bm{L}(G_S)$, $\bm{L}^{(I)}=\bm{L}(G_I)$, and $\bm{L}^{(J)}=\bm{L}(G_J)$, defined according to Definition~\ref{def:laplacian}. We also denote the entries of $\bm{L}^{(S)}$, $\bm{L}^{(I)}$, and $\bm{L}^{(J)}$ in row $i$ and column $j$ to be $L^{(S)}_{ij}$, $L^{(I)}_{ij}$, and $L^{(J)}_{ij}$.

Then, we propose the following reaction-diffusion model: 

\begin{equation}\label{eq:superinfect_model}
\begin{aligned}
    \frac{dS_i}{dt} &= rS_i\left(1-\frac{S_i}{K}\right)\left(\frac{S_i}{A}-1\right)-\frac{(\beta_1I_i+\beta_2 J_i)S_i}{S_i+I_i+J_i}\\
    &\quad+\gamma_1 I_i+\gamma_2 J_i-\mu S\\
    &\quad+d_{11}\sum_{j=1}^NL_{ij}^{(S)}S_j+d_{12}\sum_{j=1}^NL_{ij}^{(I)}I_j+d_{13}\sum_{j=1}^NL_{ij}^{(J)}J_j, \\
\frac{dI_i}{dt} &= I_i\left(\frac{\beta_1 S_i}{S_i+I_i+J_i}-\mu-\alpha_1-\gamma_1-\frac{\sigma\beta_2 J_i}{S_i+I_i+J_i}\right)\\
&\quad+d_{22}\sum_{j=1}^NL_{ij}^{(I)}I_j, \\
\frac{dJ_i}{dt} &= J_i\left(\frac{\beta_2S_i}{S_i+I_i+J_i}-\mu-\alpha_2-\gamma_2+\frac{\sigma\beta_2I_i}{S_i+I_i+J_i}\right)\\
&\quad+d_{33}\sum_{j=1}^NL_{ij}^{(J)}J_j.
\end{aligned}
\end{equation}

We include the cross-diffusion terms $d_{12}\sum_{j=1}^NL_{ij}^{(I)}I_j$ and $d_{13}\sum_{j=1}^NL_{ij}^{(J)}J_j$ to indicate how the susceptible population moves in response to infected population densities. In particular, when $d_{12}$ is positive (resp. negative), susceptible individuals gravitate toward areas with a low (resp. high) density of individuals infected with pathogen $1$, and when $d_{13}$ is positive (resp. negative), susceptible individuals gravitate toward areas with a low (resp. high) density of individuals infected with pathogen $2$.

\subsection{The Co-Infection Models (MBRD-CI)}
\label{sec:coinfect-model}
The co-infection SIS model proposed by Gao \emph{et al.} in 2016~\cite{gao2016coinfection} is
\begin{align*}
    \frac{dS}{dt}&=\mu-(\lambda_1+\lambda_2+\lambda_{12\rightarrow 1}+\lambda_{12\rightarrow 12}+\lambda_{12\rightarrow 2})S\\
    &\quad+(\gamma_1I_1+\gamma_2I_2)-\mu S,\\
    \frac{dI_1}{dt}&=(\lambda_1+\lambda_{12\rightarrow 1})S-(\lambda_2+\lambda_{12\rightarrow 2}+\lambda_{12\rightarrow 12})I_1\\
    &\quad+(\gamma_2 I_{12}-\gamma_1 I_1)-\mu I_1,\\
    \frac{dI_2}{dt}&=(\lambda_2+\lambda_{12\rightarrow 2})S-(\lambda_1+\lambda_{12\rightarrow 1}+\lambda_{12\rightarrow 12})I_2\\
    &\quad+(\gamma_1I_{12}-\gamma_2I_2)-\mu I_2,\\
    \frac{I_{12}}{dt}&=\lambda_{12\rightarrow 12}S+(\lambda_2+\lambda_{12\rightarrow 2}+\lambda_{12\rightarrow 12})I_1\\
    &\quad+(\lambda_1+\lambda_{12\rightarrow 1}+\lambda_{12\rightarrow 12})I_2-(\gamma_1+\gamma_2)I_{12}-\mu I_{12},\\
    1&=S+I_1+I_2+I_{12},
\end{align*}
where $\lambda_1\coloneqq\beta_1 I_1$, $\lambda_2\coloneqq\beta_2 I_2$, $\lambda_{12\rightarrow 12}\coloneqq\beta_{12}I_{12}$, $\lambda_{12\rightarrow 1}\coloneqq\beta_{10}I_{12}$, and $\lambda_{12\rightarrow 2}\coloneqq\beta_{02}I_{12}$.

Note that the dynamics of the two pathogens can be described in the following four ways~\cite{gao2016coinfection}:

\begin{itemize}
    \item In non-interaction transmission, the presence of the two diseases do not affect each other. In the model, this translates to $\beta_{12}+\beta_{10}=\beta_1$ and $\beta_{12}+\beta_{02}=\beta_2$.
    \item Mutual enhancement occurs when
    $\beta_{12}+\beta_{10}>\beta_1$ and $\beta_{12}+\beta_{02}>\beta_2$.
    \item The enhancement of one pathogen and inhibition of the other occurs when either
    $\beta_{12}+\beta_{10}>\beta_1$ and $\beta_{12}+\beta_{02}<\beta_2$, or\\
    $\beta_{12}+\beta_{10}<\beta_1$ and $\beta_{12}+\beta_{02}>\beta_2$.
    \item Mutual inhibition occurs when
    $\beta_{12}+\beta_{10}<\beta_1$ and $\beta_{12}+\beta_{02}<\beta_2$.
\end{itemize}

We define $I$ (resp. $J$) to represent the total $I$-infected (resp. $J$-infected) density, including both mono- and co-infections, and use $C$ in place of $I_{12}$. We combine mono- and co-infections because we assume that the two pathogens exhibit independent diffusion dynamics, and we want to consider the diffusion of the population densities associated with both mono- and co-infections. After making these modifications, we have

\begin{align*}
    \frac{dS}{dt}&=B-(\beta_1 I+\beta_2 J)S\\
    &\quad -(\beta_{10}+\beta_{02}+\beta_{12}-\beta_1-\beta_2)CS\\
    &\quad+\gamma_1I+\gamma_2J-(\gamma_{1}+\gamma_{2})C-\mu S,\\
    \frac{d I}{dt}&=[\beta_1I+(\beta_{10}+\beta_{12}-\beta_1)C](S+J-C)-\gamma_1I\\
    &\quad-\alpha_1(I-C)-\alpha_{12}C-\mu I,\\
    \frac{d J}{dt}&=[\beta_2J+(\beta_{02}+\beta_{12}-\beta_2)C](S+I-C)-\gamma_2J\\
    &\quad-\alpha_2(J-C)-\alpha_{12}C-\mu J,\\
    \frac{d C}{dt}&=\beta_{12}CS+[\beta_2J+(\beta_{02}+\beta_{12}-\beta_2)C](I-C)\\
    &\quad+[\beta_1I+(\beta_{10}+\beta_{12}-\beta_1)C](J-C)\\
    &\quad-(\gamma_{1}+\gamma_{2}+\alpha_{12})C-\mu C,\\
    1&=S+I+J-C.
\end{align*}

\begin{figure}[htbp]
    \centering
    \includegraphics[width=0.4\textwidth]{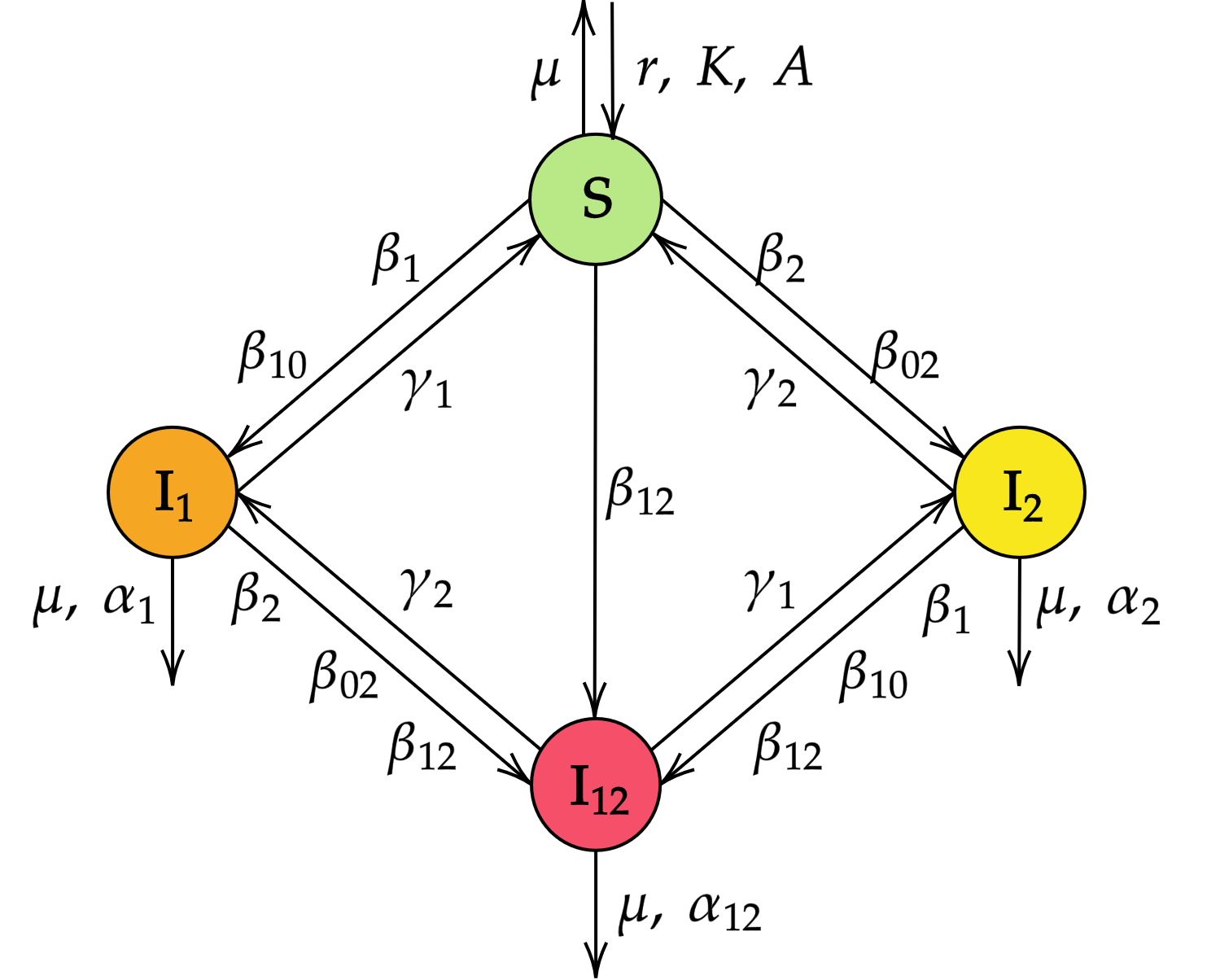}
    \caption{Flowchart for the co-infection model (MBRD-CI).}
    \label{fig:coinfect_flowchart}
\end{figure}

\begin{figure}[htbp]
    \centering
    \includegraphics[width=0.4\textwidth]{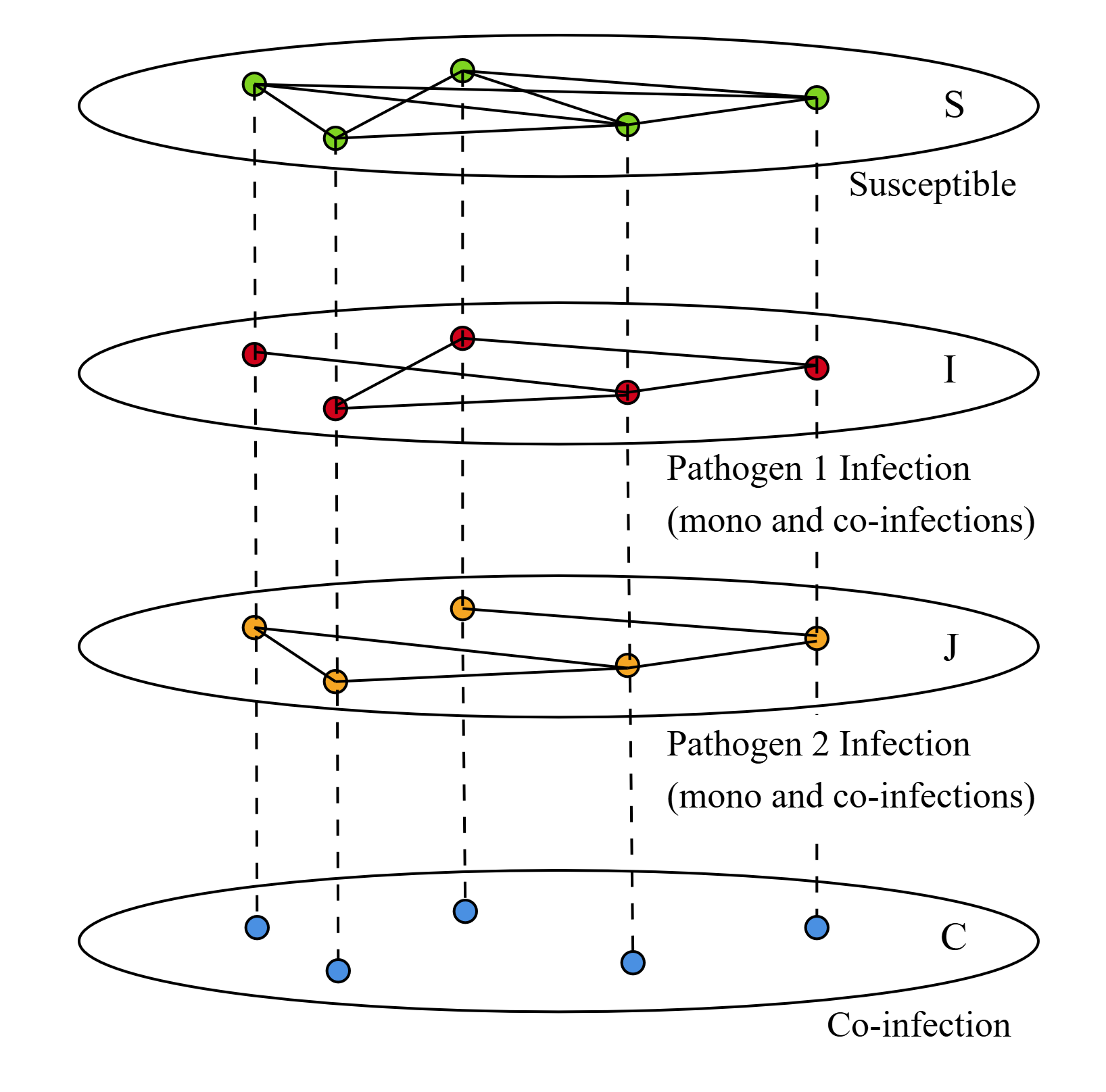}
    \caption{Four-layer multiplex network for MBRD-CI.}
    \label{fig:multiplex4layer}
\end{figure}

We modify this model so that $S$, $I$, $J$, and $C$ account for either the number of individuals or population densities. Moreover, we adopt a logistic growth framework instead of a constant birth rate, as follows:

\begin{align*}
    \frac{dS}{dt}&=rS\left(1-\frac{S}{K}\right)\left(\frac{S}{A}-1\right)-\frac{(\beta_1 I+\beta_2 J)S}{S+I+J-C}\\
    &\quad-\frac{(\beta_{10}+\beta_{02}+\beta_{12}-\beta_1-\beta_2)CS}{S+I+J-C}\\
    &\quad+\gamma_1I+\gamma_2J-(\gamma_{1}+\gamma_{2})C-\mu S,\\
    \frac{d I}{dt}&=[\beta_1I+(\beta_{10}+\beta_{12}-\beta_1)C]\cdot \frac{S+J-C}{S+I+J-C}-\gamma_1I\\
    &\quad-\alpha_1(I-C)-\alpha_{12}C-\mu I,\\
    \frac{d J}{dt}&=[\beta_2J+(\beta_{02}+\beta_{12}-\beta_2)C]\cdot\frac{S+I-C}{S+I+J-C}-\gamma_2J\\
    &\quad-\alpha_2(J-C)-\alpha_{12}C-\mu J,\\
    \frac{d C}{dt}&=\frac{\beta_{12}CS}{S+I+J-C}+\frac{[\beta_2J+(\beta_{02}+\beta_{12}-\beta_2)C](I-C)}{S+I+J-C}\\
    &\quad+\frac{[\beta_1I+(\beta_{10}+\beta_{12}-\beta_1)C](J-C)}{S+I+J-C}\\
    &\quad-(\gamma_{1}+\gamma_{2}+\alpha_{12})C-\mu C,\\
\end{align*}

Recall that our model has four states: $S$, $I_1$, $I_2$, and $I_{12}$, which are the susceptible, pathogen $1$ mono-infected, pathogen $2$ mono-infected, and co-infected densities. This system can be represented by Figure~\ref{fig:coinfect_flowchart}. The susceptible population has birth and natural death rates represented by $r\left(1-\frac{S}{K}\right)\left(\frac{S}{A}-1\right)$ and $\mu S$, respectively. The movement of individuals from $S$ to the mono-infected states and from the mono-infected states to $I_{12}$ are controlled by the infection rates. The movement of individuals from $I_{12}$ back to a mono-infected state, as well as from a mono-infected state back to $S$, are controlled by the recovery rates. Finally, deaths in the $I_1$, $I_2$, and $I_{12}$ populations include both natural and infection-related deaths.

We introduce the four-layer multiplex network in Figure~\ref{fig:multiplex4layer}. For the most part, we assume again that the average degrees of the $I$ and $J$ layers are less than the average degree of the $S$ layer. Only the $S$, $I$, and $J$ layers experience diffusion and thus, edges are not included in the $C$ layer. This is because the relative density for each node on the $C$ layer can directly be calculated from the densities of the corresponding nodes on the other three layers.

We treat the $S$, $I$, and $J$ populations as morphogens and let $S_i$, $I_i$, $J_i$, and $C_i$ be the densities of their corresponding populations on node $i$. Letting $G_S$, $G_I$, and $G_J$ be the networks on layers $S$, $I$, and $J$, we establish the same definitions for $\bm{L}^{(S)}$, $\bm{L}^{(I)}$, and $\bm{L}^{(J)}$, and their respective entries $L^{(S)}_{ij}$, $L^{(I)}_{ij}$, and $L^{(J)}_{ij}$ as our superinfection model in Equation~\ref{eq:superinfect_model}.

Then, we have
\begin{equation}\label{eq:coinfect_model}
\begin{aligned}
    \frac{dS_i}{dt}&=rS_i\left(1-\frac{S_i}{K}\right)\left(\frac{S_i}{A}-1\right)-\frac{(\beta_1 I_i+\beta_2 J_i)S_i}{S_i+I_i+J_i-C_i}\\
    &\quad-\frac{(\beta_{10}+\beta_{02}+\beta_{12}-\beta_1-\beta_2)C_iS_i}{S_i+I_i+J_i-C_i}\\
    &\quad+\gamma_1I_i+\gamma_2J_i-(\gamma_{1}+\gamma_{2})C_i-\mu S_i\\
    &\quad+d_{11}\sum_{j=1}^NL_{ij}^{(S)}S_j+d_{12}\sum_{j=1}^NL_{ij}^{(I)}I_j+d_{13}\sum_{j=1}^NL_{ij}^{(J)}J_j,\\
    \frac{d I_i}{dt}&=[\beta_1I_i+(\beta_{10}+\beta_{12}-\beta_1)C_i]\cdot \frac{S_i+J_i-C_i}{S_i+I_i+J_i-C_i}\\
    &\quad-\gamma_1I_i-\alpha_1(I_i-C_i)-\alpha_{12}C_i-\mu I_i\\
    &\quad+d_{22}\sum_{j=1}^NL_{ij}^{(I)}I_j,\\
    \frac{d J_i}{dt}&=[\beta_2J_i+(\beta_{02}+\beta_{12}-\beta_2)C_i]\cdot\frac{S_i+I_i-C_i}{S_i+I_i+J_i-C_i}\\
    &\quad-\gamma_2J_i-\alpha_2(J_i-C_i)-\alpha_{12}C_i-\mu J_i\\
    &\quad+d_{33}\sum_{j=1}^NL_{ij}^{(J)}J_j,\\
        \end{aligned}
    \end{equation}
    \begin{equation}
    \begin{aligned}\nonumber
    \frac{d C_i}{dt}&=\frac{\beta_{12}C_iS_i}{S_i+I_i+J_i-C_i}\\
    &\quad+\frac{[\beta_2J_i+(\beta_{02}+\beta_{12}-\beta_2)C_i](I_i-C_i)}{S_i+I_i+J_i-C_i}\\
    &\quad+\frac{[\beta_1I_i+(\beta_{10}+\beta_{12}-\beta_1)C_i](J_i-C_i)}{S_i+I_i+J_i-C_i}\\
    &\quad-(\gamma_{1}+\gamma_{2}+\alpha_{12})C_i-\mu C_i.
\end{aligned}
\end{equation}

\section{Three-State Instability Analysis}\label{sec:instability}

In this section, we perform an instability analysis for reaction-diffusion models with three morphogens on networks. In particular, we first derive general instability conditions for such models on a three-layer multiplex network. Then, we establish additional conditions for a special case where the layers of the multiplex network are identical. The conditions discussed in this section apply to the three-layer model in Equation~\ref{eq:superinfect_model}.

\subsection{Instability Analysis on a Three-Layer Multiplex Network}\label{sec:instab-multiplex}

We now derive the conditions for Turing and Turing-Hopf instability in a reaction-diffusion system for three distinct morphogens on a three-layer multiplex network. We consider the following system with morphogens $S$, $I$, and $J$, where $S_i$, $I_i$, and $J_i$ are the densities of the morphogens in each node of the network.
\begin{align*}
    \frac{d S_i}{dt}&=f(S_i,I_i,J_i)\\
    &\quad+d_{11}\sum_{j=1}^nL_{ij}^{(S)}S_j+d_{12}\sum_{j=1}^nL_{ij}^{(I)}I_j+d_{13}\sum_{j=1}^nL_{ij}^{(J)}J_j,\\
    \frac{d I_i}{dt}&=g(S_i,I_i,J_i)+d_{22}\sum_{j=1}^nL_{ij}^{(I)}I_j,\\
    \frac{d J_i}{dt}&=h(S_i,I_i,J_i)+d_{33}\sum_{j=1}^nL_{ij}^{(J)}J_j.
\end{align*}

Let $(S^*,I^*,J^*)$ be the steady state densities on all nodes. We define $f_S$ to be $\frac{\partial f}{\partial S}|_{(S^*,I^*,J^*)}$ and $f_I$, $f_J$, $g_S$, $g_I$, $g_J$, $h_S$, $h_I$, and $h_J$ similarly. We introduce a perturbation $\left(\delta S_i,\delta I_i,\delta J_i\right)$ to the equilibrium densities. Then, by multinomial Taylor expansions, we have
\begin{equation}\label{eq:perturb1}
\begin{aligned}
    \frac{d \delta S_i}{dt}&=f_S\delta S_i+f_I\delta I_i+f_J\delta J_i+d_{11}\sum_{j=1}^nL_{ij}^{(S)}\delta S_j\\&\quad +d_{12}\sum_{j=1}^nL_{ij}^{(I)}\delta I_j+d_{13}\sum_{j=1}^nL_{ij}^{(J)}\delta J_j,\\
    \frac{d \delta I_i}{dt}&=g_S\delta S_i+g_I\delta I_i+g_J\delta J_i+d_{22}\sum_{j=1}^nL_{ij}^{(I)}\delta I_j\\
    \frac{d \delta J_i}{dt}&=h_S\delta S_i+h_I\delta I_i+h_J\delta J_i+d_{33}\sum_{j=1}^nL_{ij}^{(J)}\delta J_j.
\end{aligned}
\end{equation}

We approximate this system as follows: 
\begin{equation}\label{eq:perturb_coinfect}
\begin{aligned}
    \frac{d \delta S_i}{dt}&=f_S\delta S_i+f_I\delta I_i+f_J\delta J_i\\
    &\quad-d_{11}k_i^{(S)}\delta S_i-d_{12}k_i^{(I)}\delta I_i-d_{13}k_i^{(J)} \delta J_i,\\
    \frac{d \delta I_i}{dt}&=g_S\delta S_i+g_I\delta I_i+g_J\delta J_i-d_{22}k_i^{(I)}\delta I_i\\
    \frac{d \delta J_i}{dt}&=h_S\delta S_i+h_I\delta I_i+h_J\delta J_i-d_{33}k_i^{(J)}\delta J_i.
\end{aligned}
\end{equation}

Letting $\bm x_i\coloneqq\left(\delta S_i, \delta I_i, \delta J_i\right)^T$, we rewrite the system in Equation~(\ref{eq:perturb_coinfect}) as
\begin{equation}\label{eq:rewritten-perturb}
    \frac{d \bm x_i}{dt}=\begin{pmatrix}
        f_S-d_{11}k_i^{(S)}&&f_I-d_{12}k_i^{(I)}&&f_J-d_{13}k_i^{(J)}\\
        g_S&& g_I-d_{22}k_i^{(I)}&&g_J\\
        h_S&&h_I&&h_J-d_{33}k_i^{(J)}
    \end{pmatrix}\bm x_i
\end{equation}
Let $\bm{M}$ be the matrix
\begin{equation*}
    \begin{pmatrix}
        f_S-d_{11}k_i^{(S)}-\lambda&&f_I-d_{12}k_i^{(I)}&&f_J-d_{13}k_i^{(J)}\\
        g_S&& g_I-d_{22}k_i^{(I)}-\lambda&&g_J\\
        h_S&&h_I&&h_J-d_{33}k_i^{(J)}-\lambda
    \end{pmatrix}.
\end{equation*}
We let $\bm x_i$ be of the form $\bm a\exp(ikx+\lambda t)$, where $\lambda$ is the growth rate. Substituting this ansatz into Equation~(\ref{eq:rewritten-perturb}), the growth rate satisfies $\det(\bm{M})=0$.

We define 
\begin{align*}
q_{11}&\coloneqq g_Ih_J-g_Jh_I, \quad & q_{22}&\coloneqq h_Jf_S-h_Sf_J,\\ 
q_{33}&\coloneqq f_Sg_I-f_Ig_S,\quad & q_{12}&\coloneqq h_Sg_J-g_Sh_J,\\ 
q_{13}&\coloneqq g_Sh_I-g_Ih_S, \quad & m_{11}&\coloneqq d_{11}k_i^{(S)},\\
m_{22}&\coloneqq d_{22}k_i^{(I)},\quad & m_{33}&\coloneqq d_{33}k_i^{(J)}, \\
m_{12}&\coloneqq d_{12}k_i^{(I)},\quad & m_{13}&\coloneqq d_{13}k_i^{(J)}.\\
\end{align*}

We also denote $e_i$ to be the $i$-th elementary symmetric polynomial and 
\begin{align*}
    p_1&\coloneqq f_S+g_I+h_J,\\
    p_2&\coloneqq q_{11}+q_{22}+q_{33},\\
    p_3&\coloneqq(f_Sg_Ih_J+f_Ig_Jh_S+f_Jg_Sh_I)\\
    &\quad-(f_Sg_Jh_I+f_Ig_Sh_J+f_Jg_Ih_S).
\end{align*}
Finally, we define
\[A(x_1,x_2,x_3,x_4,x_5)\coloneqq x_1f_S+x_2g_I+x_3h_J+x_4g_S+x_5h_S,\] and \[B(x_1,x_2,x_3,x_4,x_5)\coloneqq x_1q_{11}+x_2q_{22}+x_3q_{33}+x_4q_{12}+x_5q_{13}.\]

Then we let 
\begin{equation}\label{eq:def-p1}
    p(\lambda)\coloneqq-\det(M)=\lambda^3-b\lambda^2+c\lambda-d=0,
\end{equation}where 
\begin{equation}\label{eq:det-coeff1}
\begin{aligned}
    b&\coloneqq p_1-e_1(m_{11},m_{22},m_{33})\\
    c&\coloneqq p_2+e_2(m_{11},m_{22},m_{33})-p_1e_1(m_{11},m_{22},m_{33})\\
    &\quad+A(m_{11},m_{22},m_{33},m_{12},m_{13})\\
    d&\coloneqq p_3-B(m_{11},m_{22},m_{33},m_{12},m_{13})-e_3(m_{11},m_{22},m_{33})\\
    &\quad +A(m_{22}m_{33},m_{33}m_{11},m_{11}m_{22},-m_{12}m_{33},-m_{13}m_{22}).
\end{aligned}
\end{equation}

We denote the solutions to the system in Equation~(\ref{eq:det-coeff1}) to be $\lambda_1$, $\lambda_2$, and $\lambda_3$, where $\Re(\lambda_1)\ge \Re(\lambda_2)\ge \Re(\lambda_3)$. We prove the following two sets of necessary instability conditions:
\begin{prop}[Boundary conditions]\label{prop:instab1} \label{prop:instab1}
We must have
\begin{align*}
    p_1&<0, &\quad p_2&>0, \\
    p_3&<0, &\quad p_1p_2&<p_3.
\end{align*}
\end{prop}
\begin{proof}
Recall the previously stated definitions of $p_1$, $p_2$, and $p_3$. At equilibrium, there is no spatial diffusion, and the characteristic polynomial $p(\lambda)$ is 
\[\lambda^3-p_1\lambda^2+p_2\lambda-p_3.\]
All roots have negative real parts because no perturbations can grow into oscillations. By Vieta's formulas, we have 
\begin{align*}
    p_1&=\lambda_1+\lambda_2+\lambda_3<0,\\
    p_2&=\lambda_1\lambda_2+\lambda_2\lambda_3+\lambda_3\lambda_1>0,\\
    p_3&=\lambda_1\lambda_2\lambda_3<0.
\end{align*}
 Moreover, $p_1p_2<p_3$ follows from the Routh-Hurwitz criterion.
\end{proof}

The following definition differentiates between Turing and Turing-Hopf instability for systems of three interacting morphogens in this context. A similar definition for continuous domains is also stated in~\cite{Piskovsky2025}.
\begin{definition}\label{def:instab-roots}
    Turing instability occurs when every eigenvalue of $M$ with a positive real part is real for every Laplacian eigenvalue $k$. Turing-Hopf instability occurs when some eigenvalues with a positive real part are not real at some Laplacian eigenvalue $k$.
\end{definition}

This leads us to the following set of instability conditions:

\begin{prop}[Instability conditions I]\label{prop:instab2}
    We denote $\Delta_3$ to be the cubic discriminant $18bcd - 4b^3d + b^2c^2 - 4c^3 - 27d^2$. For Turing instability to occur, we have $c<0$ must be true under the condition that $\Delta_3=18bcd - 4b^3d + b^2c^2 - 4c^3 - 27d^2>0$ and both $c>0$ and $d>0$ must both be true under the condition that $\Delta_3<0$. For Turing-Hopf instability to occur, we must have $\Delta_3<0$ and $d<0$.
\end{prop}
\begin{proof}
    Recall the definitions from Equations~(\ref{eq:def-p1}) and~(\ref{eq:det-coeff1}). Either all roots of $p(\lambda)$ are real or there is one real root and two complex roots.
    
    If the discriminant $\Delta_3 = 18bcd - 4b^3d + b^2c^2 - 4c^3 - 27d^2$ is greater than $0$, we have three distinct real roots, which we call $x_1$, $y_1$, and $z_1$. We assume without loss of generality that $x_1>y_1>z_1$. For any spatial instability to occur, there must be an nonzero number of positive roots because a perturbation must grow into an oscillation. Recall that $b=p_1-e_1(m_{11},m_{22},m_{33})$ and $m_{11}$, $m_{22}$, and $m_{33}$ are positive by definition. Then, because $p_1<0$ by Proposition~\ref{prop:instab1}, we have $b<0$, we have that $p(\lambda)$ must have least one negative root and $x_1<-(y_1+z_1)$. First, if $x_1>0$ and $0>y_1>z_1$, it follows that $c=z_1(x_1+y_1)+x_1y_1<x_1y_1-(x_1+y_1)^2<0$ and $d=x_1y_1z_1>0$. Second, if $x_1>y_1>0$ and $z_1<0$, then $c=x_1(y_1+z_1)+y_1z_1<y_1z_1-(y_1+z_1)^2<0$ and $d=x_1y_1z_1<0$. Thus, $c<0$ must be true if $\Delta>0$, assuming that spatial oscillations occur.

    If the discriminant $\Delta = 18bcd - 4b^3d + b^2c^2 - 4c^3 - 27d^2$ is less than $0$, there exist nonreal roots. Let these roots be $x_2+y_2i$, $x_2-y_2i$, and $z_2$, where $x_2$ and $y_2$ are positive, and $z_2$ are real numbers. By Vieta's formulas, we have $b=2x_2+z_2$, $c=x_2^2+y_2^2+2x_2z_2$, and $d=z_2(x_2^2+y_2^2)$. Recall that $b$ is always negative. For spatial oscillations (Turing or Turing-Hopf) to occur, at least one root must have a positive real part. Thus either $x_2>0$ and $z_2<0$, or $x_2<0$ and $z_2>0$. First, if $z_2<-\frac{x_2^2+y_2^2}{2x_2}$, we have $c<0$ and $d<0$. Second, when $-\frac{x_2^2+y_2^2}{2x_2}<z_2<0$, we have $c>0$ and $d<0$. Third, when $z_2>0$, we have $c>0$ and $d>0$. 

    It directly follows from Definition~\ref{def:instab-roots} that Turing instability can only occur when all roots are real or there are two complex roots with negative real parts. Turing-Hopf instability occurs when two roots are complex with positive real parts. The theorem statement thus directly follows from this definition and the analysis above.
\end{proof}

\subsection{Instability Analysis on a Single-Layer Network}\label{sec:instab-single}

We shall derive additional instability conditions for the special case of Subsection~\ref{sec:instab-multiplex} where all layers are identical, which collapses to a single layer network.

We denote $G_A$ to be the single layer network, and $\bm{L_A}=\bm{L}(G_A)$ as defined in Definition~\ref{def:laplacian}. Inspired by~\cite{othmer1971instability}, we express the perturbations $\left(\delta S_i,\delta I_i,\delta J_i\right)$ as \[\left(\sum_{v=1}^Nc_v^1e^{\lambda_v t}\phi_i^{(v)},\sum_{v=1}^Nc_v^2e^{\lambda_v t}\phi_i^{(v)},\sum_{v=1}^Nc_v^3e^{\lambda_v t}\phi_i^{(v)}\right),\] where $\mu_v$ is the $v$-th eigenvalue of $\bm{L}_A$ with corresponding eigenvector $\phi_v=\left(\phi_1^{(v)},\dots,\phi_N^{(v)}\right)^T$, and $\lambda_v$ is the growth rate of the $v$-th spatial mode. 

We let 
\begin{equation}\label{eq:matrix-def1}
\bm{N}\coloneqq\begin{pmatrix}
        f_S+d_{11}\mu_v&& f_I+d_{12}\mu_v&&f_J+d_{13}\mu_v\\g_S&&g_I+d_{22}\mu_v&&g_J\\
        h_S&&h_I&&h_J+d_{33}\mu_v
    \end{pmatrix}.\end{equation}

When we substitute the ansatzes in Equation~(\ref{eq:matrix-def1}) into the system in Equation~(\ref{eq:perturb1}), we have
\begin{equation*}
    \lambda_v\bm{y}_v=
    \bm{N}\bm{y}_v,
\end{equation*}
where $\bm{y}_v\coloneqq \begin{pmatrix}c_v^1&c_v^2&c_v^3\end{pmatrix}^T$.

We let $\bm{N}$ be the matrix
\[\begin{pmatrix}
        f_S+d_{11}\mu_v&& f_I+d_{12}\mu_v&&f_J+d_{13}\mu_v\\g_S&&g_I+d_{22}\mu_v&&g_J\\
        h_S&&h_I&&h_J+d_{33}\mu_v
    \end{pmatrix}.\]

Thus, the eigenvalue $\mu_v$ of $\bm{L}_A$ and eigenvalue $\lambda_v$ of $\bm{N}$ for node $v$ satisfy
\[\det(\bm{N})=0.\]

The characteristic polynomial is 
\begin{equation}\label{eq:char-poly-single}
    \lambda_v^3-b_1(\mu_v)\lambda_v^2+c_1(\mu_v)\lambda_v-d_1(\mu_v)=0,
\end{equation} 
where
\begin{equation}\label{eq:char-def-single}
\begin{aligned}
    b_1(\mu_v)&\coloneqq p_1+e_1(d_{11},d_{22},d_{33})\mu_v,\\
    c_1(\mu_v)&\coloneqq p_2\\
    &\quad+\big[p_1e_1(d_{11},d_{22},d_{33})-A(d_{11},d_{22},d_{33},d_{12},d_{13})\big]\mu_v\\
    &\quad+e_2(d_{11},d_{22},d_{33})\mu_v^2,\\\
    &\quad\\
    d_1(\mu_v)&\coloneqq p_3+B(d_{11},d_{22},d_{33},d_{12},d_{13})\mu_v\\
    &\quad+A(d_{11}d_{22},d_{22}d_{33},d_{33}d_{11},-d_{12}d_{33},-d_{13}d_{22})\mu_v^2\\
    &\quad+e_3(d_{11},d_{22},d_{33})\mu_v^3.
\end{aligned}
\end{equation}

For simplicity, let 
\begin{equation}\label{eq:single-def1}
\begin{aligned}
A_1&\coloneqq A(d_{11},d_{22},d_{33},d_{12},d_{13}),\\
A_2&\coloneqq A(d_{22}d_{33},d_{33}d_{11},d_{11}d_{22},-d_{12}d_{33},-d_{13}d_{22}),\\
B_1&\coloneqq B(d_{11},d_{22},d_{33},d_{12},d_{13}).
\end{aligned}
\end{equation}
Then, we denote
\begin{equation}\label{eq:single-def2}
\begin{aligned}
b_0 &\coloneqq p_2, \quad & b_1 &\coloneqq A_1-p_1e_1(d_{11},d_{22},d_{33}), \\
b_2 &\coloneqq e_2(d_{11},d_{22},d_{33}), \quad & \tilde{a}_0 &\coloneqq  p_3, \\
\tilde{a}_1 &\coloneqq -B_1, \quad & \tilde{a}_2 &\coloneqq A_2, \\
\tilde{a}_3 &\coloneqq -e_3(d_{11},d_{22},d_{33}).
\end{aligned}
\end{equation}

Finally, note that $c_2(\phi_v)c_1(\phi_v)-c_0(\phi_v)=a_3\phi_v^3+a_2\phi_v^2+a_1\phi_v+a_0$, where 
\begin{equation}\label{eq:single-def3}
\begin{aligned}
    a_0&\coloneqq p_1p_2-p_3\\
    a_1&\coloneqq p_1A_1+B_1-p_2e_1(d_{11},d_{22},d_{33})-p_1^2e_1(d_{11},d_{22},d_{33})\\
    a_2&\coloneqq p_1e_2(d_{11},d_{22},d_{33})+p_1e_1^2(d_{11},d_{22},d_{33})\\
    &\quad-e_1(d_{11},d_{22},d_{33})A_1-A_2\\
    a_3&\coloneqq e_3(d_{11},d_{22},d_{33})\\
    &\quad-e_1(d_{11},d_{22},d_{33})e_2(d_{11},d_{22},d_{33}).
\end{aligned}
\end{equation}

In this scenario, Propositions~\ref{prop:instab1} and~\ref{prop:instab2} still hold. We prove the following proposition, which holds specifically for the case where the multiplex network layers are identical.

\begin{prop}[Instability conditions II]\label{prop:instab-3}
Consider the following sets of inequalities:
\begin{equation}\label{eq:polynom_cond1}
\begin{aligned}
    0&<a_2^2-3a_1a_3,\\
    0&<a_2+\sqrt{a_2^2-3a_1a_3},\\
    0&<2a_2^3+2(a_2^2-3a_1a_3)^{3/2}-9a_1a_2a_3+27a_0a_3^2,\\
\end{aligned}
\end{equation}
and
\begin{equation}\label{eq:polynom_cond2}
\begin{aligned}
    b_1&<-\sqrt{4b_2b_0},\\
    3a_3(b_1+\sqrt{b_1^2-4b_2b_0})&\leq2b_0(a_2+\sqrt{a_2^2-3a_1a_3}),\\
    2b_0(a_2+\sqrt{a_2^2-3a_1a_3})&\leq3a_3(b_1-\sqrt{b_1^2-4b_2b_0}),\\
    g\left(\frac{-b_1-\sqrt{b_1^2-4b_2b_0}}{2b_2}\right)&\leq0,\\
    g\left(\frac{-b_1+\sqrt{b_1^2-4b_2b_0}}{2b_2}\right)&\leq0,\\
\end{aligned}
\end{equation}
where $g(y)\coloneqq b_2y^2+b_1y+b_0$.

    A Turing-Hopf instability in the system defined above occurs if and only if all inequalities in the system represented by Equation~(\ref{eq:polynom_cond1}) are satisfied and at least one inequality in the system represented by Equation~(\ref{eq:polynom_cond2}) is not satisfied. A Turing instability occurs if and only if a Turing-Hopf instability does not occur and all inequalities in the system represented by Equation~(\ref{eq:polynom_cond1}) are satisfied.
\end{prop}
\begin{proof} Recall the definitions in Equations~(\ref{eq:char-poly-single}), (\ref{eq:char-def-single}), (\ref{eq:single-def1}), (\ref{eq:single-def2}), and (\ref{eq:single-def3}).
Let $\phi_v\coloneqq-\mu_v$ for every $v$. Then the characteristic polynomial is equivalent to $\lambda_v^3-c_2(\phi_v)\lambda_v^2+c_1(\phi_v)\lambda_v-c_0(\phi_v)=0$, where 
 \begin{align*}
 c_2(\phi_v)&\coloneqq -e_1(d_{11},d_{22},d_{33})\phi_v+p_1,\\
c_1(\phi_v)&= b_2\phi_v^2+b_1\phi_v+b_0,\\
c_0(\phi_v)&= \tilde{a_3}\phi_v^3+\tilde{a_2}\phi_v^2+\tilde{a_1}\phi_v+\tilde{a_0}.
\end{align*}
It is well known that the eigenvalues $\mu_v$ are all nonpositive and $0\in \{\mu_1,\dots,\mu_N\}$. Thus $\phi_v$ is interchangable with $k^2$ in~\cite{Piskovsky2025}. Verifying that the other assumptions used in~\cite{Piskovsky2025} on the coefficients of $c_2(\phi_v)$, $c_1(\phi_v)$, and $c_0(\phi_v)$ are all true for our definitions above, we conclude that Proposition~\ref{prop:instab-3} follows from Theorem~$1$ in~\cite{Piskovsky2025}.
\end{proof}

\section{Four-State Instability Analysis}\label{sec:four-state}

We now derive the conditions for Turing and Turing-Hopf instability in a reaction-diffusion system for four distinct morphogens on a four-layer multiplex network, with diffusion occurring on only three layers. We consider the following system with morphogens $S$, $I$, $J$, and $C$, where $S_i$, $I_i$, $J_i$, and $C_i$ are the densities of the morphogens in each node of the network. Because we consider diffusion on only the first three layers, the conditions discussed below apply to the MBRD-CI model in Equation~\ref{eq:coinfect_model}.

\begin{align*}
    \frac{d S_i}{dt}&=f(S_i,I_i,J_i,C_i)\\
    &\quad+d_{11}\sum_{j=1}^nL_{ij}^{(S)}S_j+d_{12}\sum_{j=1}^nL_{ij}^{(I)}I_j+d_{13}\sum_{j=1}^nL_{ij}^{(J)}J_j,\\
    \frac{d I_i}{dt}&=g(S_i,I_i,J_i,C_i)+d_{22}\sum_{j=1}^nL_{ij}^{(I)}I_j,\\
    \frac{d J_i}{dt}&=h(S_i,I_i,J_i,C_i)+d_{33}\sum_{j=1}^nL_{ij}^{(J)}J_j,\\
    \frac{d C_i}{dt} &= l(S_i,I_i,J_i,C_i).
    .
\end{align*}

Let $(S^*,I^*,J^*,C^*)$ be the steady state densities on all nodes. We define $f_S$ to be $\frac{\partial f}{\partial S}|_{(S^*,I^*,J^*,C^*)}$ and $f_I$, $f_J$, $f_C$, $g_S$, $g_I$, $g_J$, $g_C$, $h_S$, $h_I$, $h_J$, $h_C$, $l_S$, $l_I$, $l_J$, and $l_C$ similarly. We introduce a perturbation $\left(\delta S_i,\delta I_i,\delta J_i\right)$ to the equilibrium densities. Then, by multinomial Taylor expansions, we have
\begin{equation}\label{eq:perturb1_coinfect}
\begin{aligned}
    \frac{d \delta S_i}{dt}&=f_S\delta S_i+f_I\delta I_i+f_J\delta J_i+f_C\delta C_i+d_{11}\sum_{j=1}^nL_{ij}^{(S)}\delta S_j\\&\quad +d_{12}\sum_{j=1}^nL_{ij}^{(I)}\delta I_j+d_{13}\sum_{j=1}^nL_{ij}^{(J)}\delta J_j,\\
    \frac{d \delta I_i}{dt}&=g_S\delta S_i+g_I\delta I_i+g_J\delta J_i+g_C\delta C_i+d_{22}\sum_{j=1}^nL_{ij}^{(I)}\delta I_j\\
    \frac{d \delta J_i}{dt}&=h_S\delta S_i+h_I\delta I_i+h_J\delta J_i+h_C\delta C_i+d_{33}\sum_{j=1}^nL_{ij}^{(J)}\delta J_j,\\
    \frac{\delta C_i}{dt}&=l_S\delta S_i+l_I\delta I_i+l_J\delta J_i+l_C\delta C_i.
\end{aligned}
\end{equation}

We approximate this system as follows: 
\begin{equation}\label{eq:perturb}
\begin{aligned}
    \frac{d \delta S_i}{dt}&=f_S\delta S_i+f_I\delta I_i+f_J\delta J_i+f_C\delta C_i\\
    &\quad-d_{11}k_i^{(S)}\delta S_i-d_{12}k_i^{(I)}\delta I_i-d_{13}k_i^{(J)} \delta J_i,\\
    \frac{d \delta I_i}{dt}&=g_S\delta S_i+g_I\delta I_i+g_J\delta J_i+g_C\delta C_i-d_{22}k_i^{(I)}\delta I_i\\
    \frac{d \delta J_i}{dt}&=h_S\delta S_i+h_I\delta I_i+h_J\delta J_i+h_C\delta C_i-d_{33}k_i^{(J)}\delta J_i,\\
    \frac{d\delta C_i}{dt}&=l_S\delta S_i+l_I\delta I_i+l_J\delta J_i+l_C\delta C_i.
\end{aligned}
\end{equation}

Letting $\bm w_i\coloneqq\left(\delta S_i, \delta I_i, \delta J_i,\delta C_i\right)^T$, we rewrite the system in Equation~(\ref{eq:perturb}) as
\begin{equation}\label{eq:rewritten-perturb_coinfect}
    \frac{d \bm w_i}{dt}=\begin{pmatrix}
        f_S-d_{11}k_i^{(S)}&&f_I-d_{12}k_i^{(I)}&&f_J-d_{13}k_i^{(J)} && f_C\\
        g_S && g_I-d_{22}k_i^{(I)}&&g_J&&g_C\\
        h_S&&h_I&&h_J-d_{33}k_i^{(J)} && h_C\\
        l_S && l_I&&l_J&&l_C
    \end{pmatrix}\bm w_i.
\end{equation}
Recall the definitions of $m_{11}$, $m_{22}$, $m_{33}$, $m_{12}$, and $m_{13}$ from the previous section. Let $\bm{P}$ be the matrix
\begin{equation*}
    \begin{pmatrix}
        f_S-m_{11}-\lambda&&f_I-m_{12}&&f_J-m_{13} && f_C\\
        g_S&& g_I-m_{22}-\lambda&&g_J && g_C\\
        h_S&&h_I&&h_J-m_{33}-\lambda && h_C\\
        l_S&&l_I&&l_J&&l_C-\lambda
    \end{pmatrix}.
\end{equation*}
We let $\bm w_i$ be of the form $\bm a\exp(ikx+\lambda t)$, where $\lambda$ is the growth rate. Substituting this ansatz into Equation~(\ref{eq:rewritten-perturb}), the growth rate satisfies $\det(\bm{P})=0$.

We define 
\begin{align*}
u_1&\coloneq m_{11}(l_Ig_Ch_J-l_Ig_Jh_C+l_Jg_Ih_C-l_Jg_Ch_I)\\
&\quad +m_{22}(l_Sf_Ch_J-l_Sf_Jh_C+l_Jf_Sh_C-l_Jf_Ch_S)\\
&\quad +m_{33}(l_Sg_If_C-l_Sf_Ig_C+l_Ig_Ch_J-l_Ig_Jh_C)\\
&\quad +m_{12}(l_Sg_Jh_C-l_Sg_Ch_J+l_Jh_Sg_C-l_Jg_Sh_C)\\
&\quad +m_{13}(l_Sh_Ig_C-l_Sg_Ih_C+l_Ih_Cg_S-l_Ig_Ch_S),\\
u_2&\coloneqq l_Sf_Cm_{22}m_{33}+l_Ig_Cm_{11}m_{33}+l_Jh_Cm_{11}m_{22},\\
u_3&\coloneqq l_S(f_Ig_C-g_If_C+f_Jh_C-f_Ch_J)\\
&\quad +l_I(g_Jh_C-g_Ch_J+g_Sh_C-g_Ch_S)\\
&\quad +l_J(g_Ch_I-g_Ih_C+f_Ch_S-f_Sh_C),\\
u_4&\coloneqq m_{11}(l_Ig_C+l_Jh_C)+m_{22}(l_Sf_C+l_Jh_C)\\
&\quad +m_{33}(l_Sf_C+l_Ig_C)-l_S(g_Cm_{12}+h_Cm_{13}),\\
u_5&\coloneqq l_Sf_C+l_Ig_C+l_Jh_C.
\end{align*}

Recall the definitions of $e_i$, $p_1$, $p_2$, and $p_3$ from the previous section. We denote 
\begin{align*}
r_1&\coloneqq l_c-p_1,\\
r_2&\coloneqq u_5+l_cp_1-p_2,\\
 r_3&\coloneqq u_3+u_4+l_Cp_2-p_3,\\
    r_4&\coloneqq \det\begin{pmatrix}
        f_S&&f_I&&f_J&&f_C\\
        g_S&&g_I&&g_J&&g_C\\
        h_S&&h_I&&h_J&&h_C\\
        l_S&&l_I&&l_J&&l_C
    \end{pmatrix}.
\end{align*}
Recall the definitions of functions $b$, $c$, and $d$ from the previous section. Then we let 
\begin{equation}\label{eq:def-p}
    r(\lambda)\coloneqq\det(\bm{P})=\lambda^4-a'\lambda^3+b'\lambda^2-c'\lambda+d'=0,
\end{equation}where 
\begin{equation}\label{eq:det-coeff}
\begin{aligned}
a'&\coloneqq l_C-b,\\
    b'&\coloneqq u_5+l_C b-c,\\
    c'&\coloneqq u_3+u_4+l_Cc-d,\\
    d'&\coloneqq p_4+u_1-u_2+l_Cd.
\end{aligned}
\end{equation}

We denote the solutions to the system in Equation~(\ref{eq:det-coeff}) to be $\lambda_1$, $\lambda_2$, and $\lambda_3$, where $\Re(\lambda_1)\ge \Re(\lambda_2)\ge \Re(\lambda_3)$. We prove the following two sets of necessary instability conditions:
\begin{prop}[Boundary conditions]\label{prop:instab1_coinfect}
For instability to occur, we must have
\begin{align*}
    r_4&>0, &\quad r_3&<0, \\
    r_2&>0, &\quad r_1&<0,\\
    r_1r_4&>r_2r_3,&\quad r_1r_2r_3&>r_3^2+r_1^2r_4.
\end{align*}
\end{prop}
\begin{proof}
Recall the previously stated definitions of $r_1$, $p_2$, $r_3$, and $r_4$. At equilibrium, there is no spatial diffusion, and the characteristic polynomial $r(\lambda)$ is 
\[\lambda^4-r_1\lambda^3+r_2\lambda^2-r_3\lambda+r_4.\]
All roots have negative real parts because no perturbations can grow into oscillations. By Vieta's formulas, we have 
\begin{align*}
    p_1&=\lambda_1+\lambda_2+\lambda_3+\lambda_4<0,\\
    p_2&=\sum_{cyc}\lambda_1\lambda_2>0,\\
    p_3&=\sum_{cyc}\lambda_1\lambda_2\lambda_3<0,\\
    p_4&=\lambda_1\lambda_2\lambda_3\lambda_4>0.
\end{align*}
 Moreover, the conditions $r_1r_4>r_2r_3$ and $r_1r_2r_3>r_3^2+r_1^2r_4$ follow from the Routh-Hurwitz criterion.
\end{proof}

\begin{prop}[Instability conditions]\label{prop:instab1_coinfect2}
    Let $\Delta_4$ be the quartic discriminant of $r(\lambda)$. If $\Delta_4>0$, then $4\lambda^3-3a'\lambda^2+2b'\lambda-c'$ must have three real roots alternating in sign, for Turing instability to occur and $d'>0$ is a necessary condition for Turing-Hopf instability.
\end{prop}
\begin{proof}
    It is well-known that there are either four distinct real roots or four distinct complex roots when $\Delta_4>0$, two real and two complex roots when $\Delta_4<0$, and duplicate roots when $\Delta_4=0$. Assuming $\Delta_4>0$, all roots must be real for Turing instability to occur. By Rolle's theorem, there must be a real value between each pair of roots where $r(\lambda)$ has slope $0$ at that point, and these values must be alternating in sign because no roots are repeated. Thus, $r'(\lambda)=4\lambda^3-3a'\lambda^2+2b'\lambda-c'$ must have three real roots alternating in sign.

    Assuming $\Delta_4>0$, we must have four complex roots for Turing-Hopf instability to occur. Let these roots be $w+xi$, $w-xi$, $y+zi$, and $y-zi$. Thus, we have $d'=(w^2+x^2)(y^2+z^2)>0$.
\end{proof}

\section{Examples of Pattern Formation}\label{sec:examples}

In Section~\ref{sec:instability} and~\ref{sec:four-state}, we theoretically analyzed the conditions for which pattern formation occurs. We dedicate the current section to experimentally verifying that pattern formation can occur in superinfection and co-infection dynamics. A more detailed experimental discussion of pattern formation and its implications will be focused on in the companion paper \cite{SchaposnikYu2025}.

We refer to the lattice network where all nodes except the boundary nodes have degree $4$ as the LA4 network, and the lattice network where most nodes have degree $12$ as the LA12 network. We consider the superinfection dynamics in Equation~\ref{eq:superinfect_model} and the co-infection dynamics in Equation~\ref{eq:coinfect_model}. With the selected parameter configurations in Example~\ref{ex:1}, the MBRD-SI model produces the patterns in Figure~\ref{fig:superinfect_sample} on a lattice multiplex network where all layers are LA12. Similarly, the configurations in Example~\ref{ex:2} incorporated into the MBRD-CI model produces the patterns in Figure~\ref{fig:coinfect_sample} where the $S$ and $I$ layers are LA12, and the $J$ layer is LA4. 
From the growth dynamics and common spotted shapes of these patterns, we believe they arise from Turing instability.

\begin{example} \label{ex:1} (Superinfection model) 
\begin{equation}\label{eq:p1-settings}
\begin{aligned}
\mu &= 0.005,\quad &r&=0.1,\quad & A&=0.1,\quad & K&=1,\\
\beta_1 &=0.3, \quad & \beta_2&=0.15,\quad & \sigma &=3,\\
\gamma_1&=0.02,\quad &\gamma_2&=0.05,\quad & \alpha_1&=0.02,\quad & \alpha_2&=0.15,\\
d_{11}&=0.1,\quad & d_{12}&=-0.2,\quad & d_{13}&=-0.2,\\
d_{22}&=0.01,\quad & d_{33}&=4.8.
\end{aligned}
\end{equation}
\end{example}

\begin{example}\label{ex:2} (Co-infection model)
\begin{equation}\label{eq:p2-settings}
\begin{aligned}
\mu &= 0.005,\quad & r&=0.1,&\quad A&=0.1,\quad& K&=1,\\
\beta_1&=0.3,\quad & \beta_2&= 0.15,\\ \beta_{10}&=0.1,\quad & \beta_{02}&=0.1,\quad & \beta_{12}&=0.05,\\
\gamma_1&=0.02,\quad & \gamma_2&=0.05,\\ \alpha_1&=0.02,\quad & \alpha_2&=0.15,\quad & \alpha_{12}&=0.1,\\
d_{11}&=0.4,\quad & d_{12}&=-0.2,\quad & d_{13}&=-0.2,\\
d_{22}&=0.01,\quad & d_{33}&=4.8.
\end{aligned}
\end{equation}
\end{example}

\begin{figure}[htbp]
    \centering
    \includegraphics[width=0.42\textwidth]{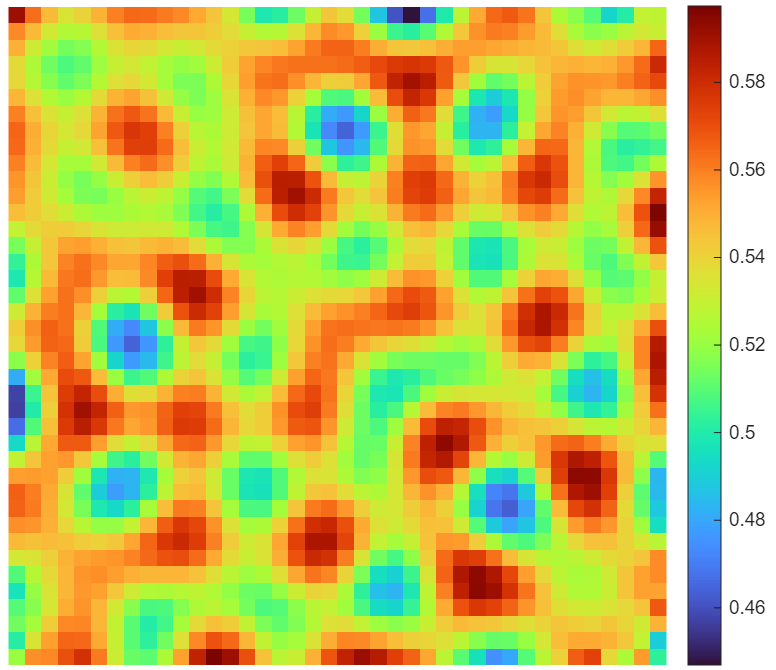}
    \caption{Pattern in superinfection dynamics, layer $I$, $t=1800$.}
    \label{fig:superinfect_sample}
\end{figure}

\begin{figure}[htbp]
    \centering
    \includegraphics[width=0.42\textwidth]{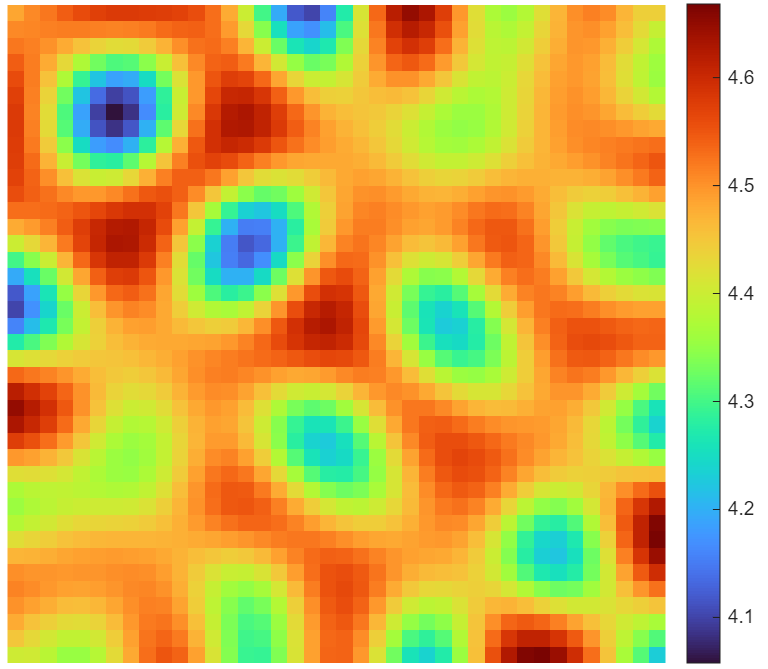}
    \caption{Pattern in co-infection dynamics, layer $I$, $t=550$.}
    \label{fig:coinfect_sample}
\end{figure}

\section{Applications}\label{sec:applications}

In this paper, we provide a framework upon which reaction-diffusion models can be created for other applications. In the following, we discuss potential applications of this framework.

\begin{itemize}
    \item \textbf{Information propagation:} The super-infection model proposed in this paper can be used to analyze the spread of conflicting or related rumors in the same network of societies. In applications of the co-infection model to social networks, individuals may be thought of as ``co-infected'' if they change their mind frequently and therefore spread both rumors.
    \item \textbf{Malware propagation:} The models in this paper can be modified to study computer viruses on networks. First, it is important to understand the dynamics between viruses and anti-viruses. Previous compartmental models have represented computers as susceptible, infected, or protected nodes~\cite{shukla2014modeling}. This could be extended into reaction-diffusion equations on a three-multiplex network, where the three layers represent the three states of computers, using the framework in this paper. Superinfection-like phenomena can occur when one particular malware is particularly dominant. Second, there is evidence that computer viruses, such as Vobfus and Beebone~\cite{microsoft}, can infect a host computer at the same time and even support one another's survival. These dynamics can be analyzed with co-infection models similar to the one presented in this paper. 
    \item \textbf{Urban planning:} Reaction-diffusion models on networks can describe how traffic congestion propagates from region to region, including how congestion in a city affects that of nearby suburbs or how freight transportation or school buses impact congestion at different times of the day. Models such as~\cite{bellocchi2020unraveling} can be modified to incorporate multiple transportation layers with road networks, commuter rail networks, or metro systems.
    \item \textbf{Election forecasting:} Compartmental epidemic models have been used to predict the 2012 and 2016 presidential elections~\cite{volkening2020forecasting}. These models can be extended to reaction-diffusion equations on networks to analyze the spatial dynamics between voting intentions of different smaller regions in the USA and other countries. An individual may be thought of as ``superinfected'' if they switch ideologies or are leaning towards one party but end up voting for another similar party that is more likely to win. Additionally, a voter can be thought of as ``co-infected" if they are moderate or believe in different aspects of two or more ideologies and are unsure of which of those parties they will vote for.
\end{itemize}

\section{Final Remarks}\label{sec:conclusion}

Over the past years, we have seen a rise in the use of reaction-diffusion dynamics to model not only epidemic spread, but also for rumor propagation and predator-prey dynamics~\cite{zhao2025navigating, ye2025pattern, song2023cross}.  In this paper, we have introduced two new deterministic frameworks: the {\bf Multiplex Bi-Virus Reaction-Diffusion models (MBRD)}. These include the  {\bf MBRD-SI} model for superinfection and the  {\bf MBRD-CI}  model for co-infection, both formulated on multiplex metapopulation networks.

Prior research has utilized stochastic processes to model superinfection and co-infection dynamics~\cite{gracy2025modeling, wu2013superinfection}. However, by integrating diffusion into our deterministic models, we capture an important characteristic of infection spread while offering computational simplicity. This makes the MBRD class of models well-suited for predicting epidemic``waves'' and large-scale pattern formation.

 To our knowledge, this is the first work to establish superinfection and co-infection reaction-diffusion epidemic models on multiplex networks. Moreover, we have derived conditions for pattern formation involving three or four morphogens, which had not previously been analyzed in a network setting.

Our MBRD-CI model, for example, could provide more accurate predictions of infections during the COVID-19 pandemic, where co-infection with influenza reached rates as high as 48\% \cite{tang2022sars}. Indeed, it can be applied to pairwise co-infections of influenza, COVID-19, and Respiratory Syncytial Virus (RSV), assuming that immunity for COVID-19 is short-lived, or co-infections of gonorrhea and chlamydia~\cite{tran2024impact}.  
On the other hand, the MBRD-SI model is well-suited for applications such as modeling HIV superinfection~\cite{redd2013frequency}, where recovery does not occur and $\gamma_1=\gamma_2=0$.

 This paper provides a foundation for which many extensions can be made. Future research building on the models introduced here could incorporate the following:
\begin{itemize}
    \item Future work can build on the current model by accounting for factors such as vaccinations, age-structuring, and cross-immunity~\cite{ram2021modified}.
    \item Extending the present model to a system with $3$, or in general $n$-pathogens, would be useful for modeling the interactions between COVID-19, influenza, and Respiratory Syncytial Virus (RSV) around the world between 2020 and 2023, among other scenarios. 
    \item In many cases, human movement between two communities may be particularly large or small, or may only be one-directional. Accounting for these differences through weighted and directed networks may produce more accurate models for predicting infectious spread.
    \item Similar models to the superinfection model introduced in this paper can be proposed for superinfection exclusion. This can be used to predict the spread of the West Nile virus and the flavivirus~\cite{goenaga2020superinfection}, among other types of infections.
    \item The current SIS model cannot be directly applied to vector-borne diseases. A vector-borne adaptation of the co-infection model in Equation~\ref{eq:coinfect_model} can be used to investigate malaria and helminth co-infections~\cite{mwangi2006malaria}, Zika and dengue co-infections~\cite{bonyah2019co}, and COVID-19 and dengue co-infections~\cite{verduyn2020co}. Vector-borne adaptations of the superinfection model proposed in Equation~\ref{eq:superinfect_model} can be used to model different strains of dengue viruses~\cite{eegunjobi2023modelling}, among others.
\end{itemize}

In this paper, we focus on introducing and theoretically analyzing new models for superinfection and co-infection. This work will be continued by an accompanying paper that uses our models to understand the impact of various factors on hotspot growth and the spread of infections across a human meta-population network.

\noindent {\bf Acknowledgements.}\\
The authors are thankful to the MIT PRIMES-USA program for their support and the opportunity to conduct this research together.
The research of LPS was partially supported by  NSF FRG Award DMS- 2152107 and an NSF CAREER Award DMS 1749013.  \\

\noindent {\bf Affiliations.}\\
(a) Poolesville High School, Poolesville,   USA.\\
(b)  University of Illinois at Chicago,  USA. \\
 
\bibliographystyle{unsrt}
\bibliography{PRIMES_2025}{}

\appendix

\end{document}